\documentclass[
 reprint,
 superscriptaddress,
 amsmath,amssymb,
 aps,
 pra,
 10pt
]{revtex4-2}

\usepackage{graphicx}
\usepackage{dcolumn}
\usepackage{bm}
\usepackage{bbold}
\usepackage{varioref}
\usepackage[
hidelinks,
colorlinks=false,
citecolor=black,
urlcolor=black,
linkcolor=black
]{hyperref}
\usepackage[capitalise]{cleveref}
\usepackage[acronym]{glossaries}
\usepackage{braket}
\usepackage{booktabs}

\usepackage{amsthm} 
\usepackage{amsmath}
\usepackage{mathtools}
\usepackage{siunitx}
\usepackage[caption=false]{subfig}
\usepackage{xcolor}

\newtheorem{theorem}{Theorem} 
\newtheorem{lemma}{Lemma}[section]

\newcommand{\IBMYRK}{IBM Quantum, IBM T.J. Watson Research Center, Yorktown Heights, NY 10598, United States}
\newcommand{\IBMZRL}{IBM Quantum, IBM Research Europe - Zurich, CH-8803 Rüschlikon, Switzerland}
\newcommand{\EPFL}{Institute of Physics, École Polytechnique Fédérale de Lausanne (EPFL), CH-1015 Lausanne, Switzerland}
\newcommand{\ETH}{Institute for Theoretical Physics, ETH Zürich, CH-8093 Zürich, Switzerland}
\newcommand{\STFC}{The Hartree Centre, STFC, Sci-Tech Daresbury, Warrington, WA4 4AD, United Kingdom}
\newcommand{\MP}{Max Planck Institute for Solid State Research, Heisenbergstr. 1, 70569 Stuttgart, Germany}
\newcommand{\CAM}{Yusuf Hamied Department of Chemistry, University of Cambridge, Lensfield Road, Cambridge CB2 1EW, United Kingdom}

\newcommand\matH{\boldsymbol{\mathrm{H}}}
\newcommand\matS{\boldsymbol{\mathrm{S}}}
\newcommand\matV{\boldsymbol{\mathrm{V}}}
\newcommand\vecm{\boldsymbol{\mathrm{m}}}

\begin{document}

\title{Quantum chemistry with provable convergence via randomized\texorpdfstring{\\sample-based Krylov quantum diagonalization}{}}

\author{Samuele Piccinelli}%
\email{samuele.piccinelli@ibm.com}
\affiliation{\IBMZRL}%
\affiliation{\EPFL}%
\author{Alberto Baiardi}%
\affiliation{\IBMZRL}%
\author{Stefano Barison}%
\affiliation{\IBMZRL}%
\affiliation{\ETH}%
\author{Max Rossmannek}%
\affiliation{\IBMZRL}%
\author{Almudena Carrera Vazquez}%
\affiliation{\IBMZRL}%
\author{Francesco Tacchino}%
\affiliation{\IBMZRL}%
\author{Stefano Mensa}%
\affiliation{\STFC}%
\author{Edoardo Altamura}%
\affiliation{\STFC}%
\affiliation{\CAM}%
\author{Ali Alavi}%
\affiliation{\MP}%
\affiliation{\CAM}%
\author{Mario Motta}%
\affiliation{\IBMYRK}%
\author{Javier Robledo-Moreno}%
\affiliation{\IBMYRK}%
\author{William Kirby}%
\affiliation{\IBMYRK}%
\author{Kunal Sharma}%
\affiliation{\IBMYRK}%
\author{Antonio Mezzacapo}%
\affiliation{\IBMYRK}%
\author{Ivano Tavernelli}%
\email{ita@zurich.ibm.com}
\affiliation{\IBMZRL}%

\date{\today}

\begin{abstract}
Quantum algorithms based on classical processing of individual samples have recently emerged as a promising route to perform large-scale quantum-chemical calculations on pre-fault-tolerant and early-fault-tolerant quantum devices.
In these algorithms, the quantum computer acts as a sampling engine that generates the subspace in which the Hamiltonian is classically diagonalized.
The recently proposed Sample-based Krylov Quantum Diagonalization (SKQD) uses quantum Krylov states as circuits from which samples are collected.
Convergence guarantees can be derived for SKQD under assumptions similar to those of quantum phase estimation -- namely, a sufficiently large overlap between the reference and the target ground state -- with the additional requirement that the ground-state wave function be well approximated by a polynomial subset of the full Hilbert space.
However, implementations of SKQD for complex many-body Hamiltonians, such as quantum chemistry ones, are limited by the depths of time-evolution circuits needed to generate Krylov vectors.
In this work, we introduce a method that combines SKQD with a qDRIFT randomized compilation of the Hamiltonian propagator.
The resulting algorithm, termed SqDRIFT, enables executing quantum chemistry experiments on present-day quantum processors while retaining provable convergence guarantees, which we extend to account for the error of the randomized compilation and of finite sampling.
We demonstrate its viability by applying SqDRIFT to calculate the electronic ground-state energy of several polycyclic aromatic hydrocarbons, up to system sizes beyond the reach of exact diagonalization.
\end{abstract}

\maketitle

\section{Introduction}

Quantum computing has emerged as a novel computational paradigm capable of addressing specific problems that exhibit unfavorable scaling on classical computers.
Quantum chemistry is a prime candidate among these problems, and quantum computers are believed to be able to speed up the solution of the electronic Schr\"odinger equation within a given basis set, which is a central computational challenge in quantum chemistry~\cite{Cao2019_QuantumChemistry-QuantumComputing-Review,McArdle2020_QuantumComputChem,Motta2022_Review,Nutzel2024}.
Quantum algorithms for solving the Schr\"odinger equation with performance guarantees already exist, and are mostly based on Quantum Phase Estimation (QPE)~\cite{kitaev1995phaseestimation,kitaev2002computation,AspuruGuzik2005_QPE-Original,Reiher2017_FeMoCo,Lee2021_TensorFactorization,Goings2022_P450}.
Since QPE requires executing deep and complex quantum circuits, beyond the reach of current noisy quantum computers, alternative approaches tailored to present-day quantum processors have been actively investigated.
A prime example are variational algorithms such as the Variational Quantum Eigensolver (VQE)~\cite{Peruzzo2014_VQE,Kandala2017_VQE-Experiment,Tilly2022_VQE-Review}, where a parameterized quantum circuit is optimized classically, according to a cost function given by evaluating the expectation value of the Hamiltonian on a quantum computer.
However, the practical implementation of VQE faces scalability challenges due to a steep measurement overhead~\cite{Wecker2015_VQE-Measurement} and the difficulty of navigating the optimization landscape~\cite{Larocca2025_BarrenPlateaus-VQE}. Alternative strategies are required to scale-up quantum-chemical calculations on quantum computers beyond the reach of brute-force classical simulations -- the so-called \emph{quantum utility} regime~\cite{Kim2023_Utility}.
The brute-force approach for solving the ground-state problem in quantum-chemical simulation corresponds to full configuration interaction (FCI); advances in FCI algorithms and in classical high-performance computing have pushed exact diagonalization to active spaces of $100$ orbitals and $\sim10^{15}$ determinants~\cite{gao2024trillion,shayit2025quadrillion}.

Methods exploiting quantum-centric supercomputing platforms to sample classically-hard probability distributions -- and classically process those samples -- are emerging as a promising route towards this goal.
Examples of this class of methods for natural science applications are quantum algorithms inspired by classical selected configuration interaction (SCI)~\cite{Evangelisti1983_CIPSI,Holmes2016_HBCI,Schriber2017_AdaptiveCI,Zimmerman2017_I-FCI,Tubman2020_SelectedCI,Chilkuri2021_SelectedCI}. The extension of classical SCI methods to the quantum setting was proposed in the quantum-selected configuration interaction method (QSCI)~\cite{Kanno2023_QuantumSelectedCI} and in the sample-based quantum diagonalization (SQD), which enabled chemistry experiments up to $72$ spin orbitals using error mitigation techniques such as configuration recovery (CR)~\cite{RobledoMoreno2024_SQD-Original}.
In these approaches, one diagonalizes a many-body Hamiltonian in a subspace of Slater determinants generated by collecting samples from a quantum circuit. Since both approaches share the same classical diagonalization step, any quantum advantage over classical SCI can originate only from the determinant-selection step, i.e.\, from the ability of the quantum circuit to identify the configurations supporting the ground state more effectively than classical selection heuristics. Whether this is the case in general remains an open question, although a recent experiment reports improved ground-state accuracy over classical sparse-solver baselines on a quantum processor~\cite{kirby2026observation}; we note, however, that the classical hardness of sampling from generic quantum circuits implies that such an advantage, if realized, could not be straightforwardly de-quantized. In either case, the approach is practical only if accurate energies can be obtained with a sufficiently small diagonalization subspace. This condition is met by ground states that are well approximated by wavefunctions supported over a small subset of the full Hilbert space.
In particular, the SQD workflow showcased complex molecular systems, with implementations successfully utilizing up to $77$ qubits~\cite{RobledoMoreno2024_SQD-Original,Barison2025_SQD-ExcitedStates,Danilov2025_SQD-AFQMC,kaliakin2024_supramolecular,Liepuoniute_SQD_triplet}.

The efficiency of SQD crucially depends on the selection of the quantum circuit used to collect the samples. Borrowing from VQE-based approaches, hardware-efficient variational circuit ans\"atze~\cite{Matsuzawa2020_UnitaryJastrow,Motta2023_LUCJ} can be used to this end. However, this introduces a heuristic component as these ans\"atze are not guaranteed to sample the support of the exact ground-state wave function.

A circuit choice that eliminates the need for a variational ansatz are time-evolution circuits~\cite{Sugisaki2024,Mikkelsen2025_QuantumSelectedCI-Krylov,Yu2025_SKQD}, which are employed in Krylov subspace diagonalization methods~\cite{parrish_2019,epperly2021subspacediagonalization,motta2023subspace,Yoshioka2025_Krylov-Experiment,Stair2023}.
We will refer to these as Sample-based Krylov Quantum Diagonalization (SKQD) methods, following~\cite{Yu2025_SKQD}, which are a direct extension of Krylov Quantum Diagonalization (KQD) methods~\cite{parrish_2019,motta2023subspace,Yoshioka2025_Krylov-Experiment}.
A complementary line of work extracts spectral information not from a configuration subspace but from a \emph{signal} subspace built from the time series of measured observables, using dynamic mode decomposition combined with classical-shadow tomography to reduce the measurement overhead~\cite{shen2026modmd}. Unlike the configuration-based methods considered here, such approaches return the spectrum rather than an explicit wave function. The key advantage of SKQD is that time-evolution circuits are guaranteed to sample the configurations on which the ground-state has a large support, provided that the initial state for the propagation has a sufficiently large overlap with the true ground state (and that the wave function is concentrated, as in any SQD- and KQD-based method).

In addition to this formal advantage, SKQD can be implemented on current state-of-the-art quantum processors~\cite{Yu2025_SKQD} provided that the time-evolution circuits can be efficiently executed on current hardware architectures.
This is not the case of generic quantum-chemical Hamiltonians, for which implementing even a single first-order Trotter step may yield circuits that are too deep to be reliably implemented on currently-available quantum computers~\cite{Mikkelsen2025_QuantumSelectedCI-Krylov} already for small-size molecular system.
Approximate compilation schemes can be employed to reduce the depth of implementing a single Trotter step.
This can be achieved, e.g.\, using matrix product operator (MPO)-based compression methods~\cite{Gibbs2025_MPO-Compression,Robertson2025_AQC-Tensor}, as proposed and experimentally demonstrated for quantum circuits of up to $36$ qubits in Ref.~\cite{Sugisaki2024}.
However, tensor-network based compression techniques can face scalability challenges when the time evolution circuits generate volume-law entanglement.
Hence, realizing SKQD on realistic quantum chemical simulations of larger molecular systems continues to pose a major challenge.

In this work, we fill this gap and present a practical, quantum-chemistry friendly variant of SKQD, which provides an alternative approach to circuit depth reduction with convergence guarantees and not limited by entanglement growth. 
Instead of compiling the time-evolution circuit based on conventional Trotter formulas~\cite{Childs2021_Trotter}, we propose to leverage the qDRIFT randomized compilation strategy~\cite{Campbell2019}.
In qDRIFT, the time-evolution is realized through an ensemble of circuits, each one obtained by selecting randomly a subset of Hamiltonian terms, with a probability proportional to the corresponding Hamiltonian coefficient.
For sparse Hamiltonians, like the quantum-chemical ones, this compilation strategy yields asymptotically shallower circuits compared to Trotter-based formulas.
Here we propose to collect the samples defining the subspace from the time-evolution circuits approximated by qDRIFT unitaries.
We refer to this approach as \textit{SqDRIFT}.
By leveraging theoretical results on the asymptotic scaling of qDRIFT~\cite{Campbell2019,Campbell2022_qDRIFT-QPE}, we extend the convergence guarantees of SKQD~\cite{Yu2025_SKQD} to SqDRIFT.
This ultimately demonstrates that SqDRIFT can efficiently sample the support of concentrated ground-state wave functions with shallower circuits.
Crucially, this construction is systematically improvable: increasing the resources invested in circuit depth and sampling provably drives the probability of missing important configurations to zero, in contrast to classical determinant-selection heuristics, which offer no such \emph{a priori} guarantee.
Besides this formal scaling analysis, we show, based on classical simulations and quantum computations, how SqDRIFT enables practical quantum-chemical SQD calculations on quantum computers.
Specifically, we calculate the ground-state energy of two polycyclic aromatic hydrocarbon (PAH) systems, namely naphthalene and coronene, using an active space that includes all $\pi$-type molecular orbitals.
Classical, noiseless simulations on naphthalene showcase that SqDRIFT can efficiently generate samples belonging to the support of the ground-state wave function.
Importantly, the sample accuracy increases with the qDRIFT circuit depth, hence indicating that SqDRIFT can trade circuit complexity for sampling overhead.
Moreover, we successfully execute the SqDRIFT protocol on state-of-the-art IBM Quantum processors for coronene, demonstrating that it can achieve accurate results on up to $48$ qubits.

\section{Theory}

\subsection{The concentration hypothesis}

Quantum sampled-subspace methods are tailored to Hamiltonian for which the ground state is well-approximated by a concentrated wave function.
\emph{Concentrated} means that the wave function is supported over a small portion of the Hilbert space.
More precisely, an $n$-qubit wave function $\ket{\Psi}$ is \emph{$(\alpha_L, \beta_L)$-concentrated}~\cite{Yu2025_SKQD} if it can be expressed as:
\begin{equation}
  \ket{\Psi} = \sum_{i=1}^{2^{n}} c_i \ket{b_i} \, ,
  \label{eq:ConcentratedWaveFunction}
\end{equation}
where $\ket{b_i}$ are ordered $n$-qubit Slater determinants (or \emph{bitstrings}) such that ${\left\| c_1 \right\|\geq \ldots\geq \left\| c_{2^n} \right\|}$, satisfying
\begin{align}
  \sum_{i=1}^L \lVert c_i\rVert^2 
    &\ge \alpha_L
  &&\text{and}\; 
  &\lVert c_i\rVert^2
    &\ge \beta_L 
  &\forall i \le L\, .
  \label{eq:ConcentrationConstraint}
\end{align}
We consider a wavefunction to be \emph{concentrated} if it is $(\alpha_L, \beta_L)$-concentrated with $L=\text{poly}(n)$ for $n$ qubits, ${1-\alpha_L\ll1}$, and $\beta_L=\Omega(1/\text{poly}(n))$: this makes precise the notion that the wavefunction is approximately supported on those $L$ bitstrings.
Whether a given system satisfies this hypothesis is a property of the target ground state and of the orbital basis, and is shared by all determinant-selection methods, quantum or classical: selected configuration interaction is efficient in precisely the same regime, and no determinant-subspace method can succeed with polynomial resources when it is violated.
Importantly, concentration is distinct from single-reference character: a strongly (multi-reference) correlated ground state can itself be concentrated, its weight shared among many near-degenerate determinants whose number still grows only polynomially.
This is exactly the regime in which determinant-selection methods, classical or quantum, are most useful -- a single Hartree-Fock reference is inadequate, yet the relevant configurations remain tractably few.
The hypothesis degrades only in the genuine strong-correlation limit, where the weight spreads over an exponentially growing number of configurations and $\beta_L$ decays rapidly with $L$~\cite{schriber2016adaptive}; a nearly single-determinant ground state, at the opposite extreme, satisfies it trivially but is better treated by coupled cluster.

Although verifying the hypothesis exactly would require the full-CI solution, several inexpensive proxies bound it from either side. Standard single-reference diagnostics -- the Hartree-Fock overlap $|\gamma_0|^2 = |\langle\mathrm{HF}|\Psi_0\rangle|^2$, the natural-orbital occupation numbers of the one-particle reduced density matrix from any affordable correlated method, and the $\mathcal{T}_1/\mathcal{D}_1$ diagnostics -- certify the easy, weakly correlated case, but a large value flags multireference character without implying that concentration fails, since the weight may still be shared among tractably few determinants.%

The bitstrings $\ket{b_i}$ yielding the largest contribution to the ground-state wave function in \cref{eq:ConcentratedWaveFunction} are identified by sampling from a quantum circuit.
The Hamiltonian is then projected and diagonalized classically in the subspace generated by the sampled bitstrings, from which one obtains an approximation of the ground state energy and of the corresponding wavefunction supported in the subspace itself. 
The concentration hypothesis guarantees that accurate approximations are reachable without employing an exponentially large subspace, i.e.\, with manageable classical computational resources.

\subsection{Sample-based Krylov Quantum Diagonalization}

We review here the SKQD workflow introduced in~\cite{Yu2025_SKQD}.
The key challenge is to construct a quantum circuit that is at the same time easy to implement on quantum processors and guaranteed to sample the important bitstrings efficiently.
Ideally, the circuit should generate a \emph{uniform} probability distribution over the important bitstrings $\ket{b_i}$ of \cref{eq:ConcentratedWaveFunction}.
This minimizes the sampling overhead needed to extract all $\ket{b_i}$ bitstrings.
Notably, a circuit representing the ground-state wave function $\ket{\phi_0}$ may not generate such a distribution, since a given bitstring $\ket{b_i}$ will be sampled with probability $\left| \langle \phi_0 \vert b_i \rangle \right|^2$, which can be far from uniform over the important bitstrings.
Hence, conventional ans\"atze tailored to encode ground-state wave functions may not yield good trial states for sample-based approaches.
For this reason, identifying a good set of non-variational trial states is desirable.
As we will discuss below, time-evolution circuits constitute such a set.
These states are used in SKQD, and KQD methods: we review both variants in the next section.

Consider an $n$-qubit Hamiltonian $H$ with eigenvalues $E_0\leq E_1\leq\ldots$ and exact ground state $|\phi_0\rangle$.
Given a reference wavefunction $|\psi_0\rangle$ and reference evolution time $t$, the Krylov subspace $\mathcal{K}_d$ of order $d$ (here and in the following, we assume $d$ odd so that an equivalent, symmetric two-sided basis centered at $k=0$ can be used in the convergence proof of \cref{appendix:proof_thm_krylov})~\cite{Saad1989_Krylov} is constructed as the linear space spanned by the vectors
\begin{equation}
  |\psi_k\rangle = \mathrm{e}^{-\mathrm{i}kHt}|\psi_0\rangle \, ,
  \label{eq:krylov_vectors}
\end{equation}
for $k\in\{0,1,\ldots,d-1\}$.
By projecting $H$ in $\mathcal{K}_d$, the task of approximating the ground state of the target Hamiltonian is then reduced to solving the generalized eigenvalue problem
\begin{equation}
  \matH v = \tilde{E}\matS v \, .
  \label{eq:GSE-Krylov}
\end{equation}
Here, $\matH_{ij} = \langle \psi_i | H |\psi_j\rangle $, $\matS_{ij} = \langle \psi_i | \psi_j \rangle$ and each vector $v$ defines a quantum state $|\phi\rangle = \sum_k v_k |\psi_k\rangle \in \mathcal{K}_d$.
The matrix elements of $\matH$ and $\matS$ can, in principle, be evaluated using a quantum computer~\cite{Yoshioka2025_Krylov-Experiment} provided that the unitary evolution operators $U_k(t)=\mathrm{e}^{-\mathrm{i}kHt}$ can be implemented with arbitrary precision. In general, however, calculating terms of the form $\langle \psi_i | H |\psi_j\rangle$ with $i\neq j$ require executing circuits containing quantum controlled operations, which are challenging to compile and realize on noisy quantum processors with local connectivity.

SKQD, the sample-based counterpart of KQD, bypasses this limitation by combining elements from both SQD and KQD: instead of forming matrix elements between the time-evolved states, one samples from them.
These time-evolved states are
\begin{equation}
  \ket{\Psi_k} = \left( \prod_{j=1}^k \mathrm{e}^{-\mathrm{i}Ht} \right)
  \ket{\Psi_\text{init}} = \mathrm{e}^{-\mathrm{i}Hkt} \ket{\Psi_\text{init}} \, ,
  \label{eq:TrotterCircuits}
\end{equation}
where the initial reference wavefunction $\ket{\Psi_\text{init}}$ should be chosen such that it is easy to prepare on a quantum computer and has a non-negligible overlap with $|\phi_0\rangle$. Samples are collected from the circuits preparing $\ket{\Psi_k}$ for different $k$ values. Then, they are merged together, and $H$ is diagonalized in the resulting subspace. Note that the bitstrings collected from the quantum circuits are split into their spin-$\alpha$ and spin-$\beta$ components, and the diagonalization subspace is spanned by all determinants formed as products of the unique configurations retained in each spin sector. Denoting by $\mathrm{d}_{\alpha/\beta}$ the number of unique configurations per spin sector, the subspace therefore has dimension $D=\mathrm{d}_\alpha\mathrm{d}_\beta$.
Remarkably, it was proven in Ref.~\cite{Yu2025_SKQD} that, if ground state concentration assumptions are met, SKQD is guaranteed to efficiently yield accurate solutions in the asymptotic limit.

To implement SKQD -- and thus benefit from its unique combination of convergence guarantees and affordable experimental requirements -- suitable unitary compilation strategies for time-evolution operators are needed.
However, while standard approaches (for instance, Suzuki-Trotter product formulas~\cite{miessen2023quantum}) work well for spin and impurity models~\cite{Yu2025_SKQD}, they become quickly impractical when targeting molecular systems, particularly on noisy quantum processors, because they yield very deep circuits.
In this work, we aim at enabling quantum chemistry SKQD calculations by adopting qDRIFT, a randomized time-evolution approach~\cite{Campbell2019}. 

\subsection{The qDRIFT compilation protocol}

The quantum stochastic drift (qDRIFT) protocol~\cite{Campbell2019} aims at compiling the time-evolution operator $\mathrm{e}^{-\mathrm{i} H t}$ for Hamiltonians expressed as
\begin{equation}
  H = \sum_{i=1}^{\mathcal{N}} c_i h_i \, ,
  \label{eq:Hamiltonian_Decomposition}
\end{equation}
where, without loss of generality, we require $c_i > 0$ and that the largest eigenvalue of $h_i$ be equal, in absolute value, to $1$.
Note that \cref{eq:Hamiltonian_Decomposition} includes, as a special case, the decomposition of $H$ in the Pauli basis (i.e.\, in terms of Pauli strings).
However, other choices are possible such as, e.g.\, mapping each $h_i$ term to a string of fermionic second-quantization operators.
The qDRIFT protocol relies on a subroutine that constructs, for a target time $t$, the unitary operator $V_{\vecm}$ defined as
\begin{equation}
  V_{\vecm} = \prod_{j=1}^N \mathrm{e}^{- \mathrm{i} h_{m_j} t \lambda / N} \, ,
  \label{eq:qDRIFT_Unitary}
\end{equation}
where $\lambda = \sum_i |c_i|$, and the series of indices $\vecm = (m_1, \ldots, m_N)$ is obtained by randomly sampling terms $h_i$ from the distribution defined by $c_i/\lambda$.
The average of the channel defined in \cref{eq:qDRIFT_Unitary} over all possible indices $\vecm$, i.e.\,
\begin{equation}
  \mathcal{E}_\text{qDRIFT}[\rho] = \sum_{\vecm} p_{\vecm} V_{\vecm} \rho V_{\vecm}^\dagger \, ,
  \label{eq:Average_qDRIFT_Channel}
\end{equation}
where $p_{\vecm} = \lambda^{-N}\prod_{i=1}^{N} c_{m_i}$ is the probability of sampling the $\vecm$ indices, yields an approximation to the exact time-evolution channel $\mathcal{U}_t[\cdot] = \mathrm{e}^{-\mathrm{i} H t}(\cdot)\mathrm{e}^{\mathrm{i} H t}$.
As proven in Ref.~\cite{Campbell2019}, the approximation error $\varepsilon_{\text{ch}}$ on the time-evolution channel \cref{eq:Average_qDRIFT_Channel} grows as $\mathcal{O}(\lambda^2 t^2/N)$.
Notably, $\varepsilon_{\text{ch}}$ does not depend on the number of terms $\mathcal{N}$ appearing in \cref{eq:Hamiltonian_Decomposition}, but rather on the $L_1$-norm of $H$.
Hence, qDRIFT yields an improvement over conventional Trotter formulas for sparse Hamiltonians, whose norm does not depend on (or changes very slowly with) the system size.

Notably, as discussed in Ref.~\cite{Huang2021_qDRIFT-Concentration}, even a single term of the sum appearing in \cref{eq:Average_qDRIFT_Channel} yields an accurate approximation of the time-evolution operator with high probability.
However, in this case the number of terms $N$ required to obtain a given target accuracy $\varepsilon_{\text{ch}}$ increases quadratically compared to the original qDRIFT protocol outlined above.

Before moving forward, we analyze convergence guarantees and error bounds for KQD implemented using qDRIFT.
These will pave the way towards our main results.
Specifically, in \cref{thm:KQD_QDRIFT} in \cref{appendix:proof_thm_krylov} we prove that if all states $\{|\psi_k\rangle\}$ defined in \cref{eq:krylov_vectors} are prepared using $N_r$ qDRIFT randomizations of length $N$, then, for any $\delta > 0$, \cref{eq:GSE-Krylov} yields an approximate ground state energy $\tilde{E}$ satisfying
\begin{equation}
        \tilde{E}-E_0 \leq \xi\,,
\end{equation}
with probability at least $1-\delta$, where
\begin{align}
    \xi = \frac{\chi}{|\gamma^\prime_0|^2} + \frac{6||H||}{|\gamma_0^\prime|^2} & \left(\frac{2\chi}{\Delta^\prime} + \zeta + 8\left(1+\frac{\pi\Delta^\prime}{4||H||}\right)^{-2d + 1}\right)
    \label{eq:xi-krylov}
\end{align}
and 
\begin{align}
    \chi&\leq 2\varepsilon_Q ||H||,\\
    \zeta &\leq 2d(\varepsilon_R + \varepsilon_Q),\\
    |\gamma^\prime_0|^2 & \geq |\gamma_0|^2 - 2\varepsilon_R - 2\varepsilon_Q\,,
\end{align}
with
\begin{equation}
        \varepsilon_Q = d(d-1)t\lambda\left(\frac{t(d-1)\lambda}{N} + \sqrt{\frac{11\ln(2^{n+1}/\delta)}{NN_r}}\right)\,.
        \label{eq:epsilonQ}
\end{equation}
Here, $\varepsilon_R$ is a regularization threshold~\cite{kirby2024analysis}, $|\gamma_0|^2$ is the overlap between $|\psi_0\rangle$ and the true ground state, ${\Delta^\prime = \Delta - \chi/|\gamma^\prime_0|^2}$ is a rescaled version of the spectral gap ${\Delta = E_1-E_0}$ and the evolution time is set to ${t=\pi / (E_{2^n - 1}-E_0)}$.
Notice that $\delta$ controls the tradeoff between the success probability ($1-\delta$) and energy error (\cref{eq:xi-krylov}).

This result shows that the quality of the solutions obtained via qDRIFT-based Krylov diagonalization can be systematically and efficiently (i.e.\, at a polynomial cost in terms of quantum resources) improved by increasing the Krylov dimension $d$ and the length of qDRIFT sequences $N$.

\subsection{The SqDRIFT algorithm}
\begin{figure*}[t]
    \centering
    \includegraphics{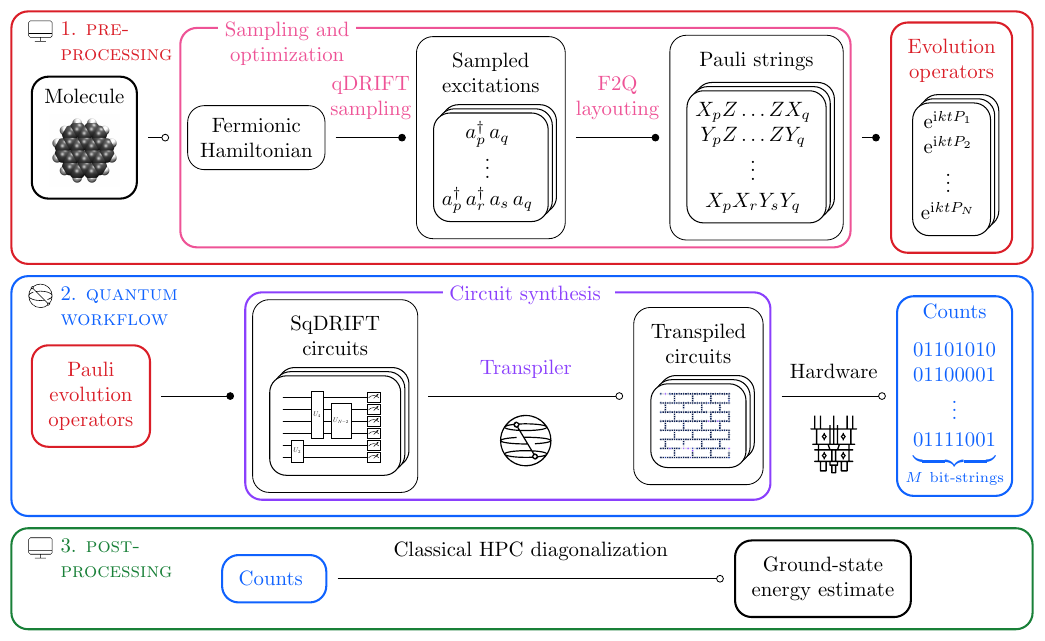}
    \caption{\textbf{SqDRIFT quantum-centric supercomputing workflow}. The diagram is to be read from top to bottom, left to right. Processes are indicated above the arrows, while results are enclosed in boxes. The classical or quantum parts of the pipeline are indicated by the top-left icons. The main contributions to the SQD pipeline are highlighted with filled circles. For details on the various steps, see the main text.}
    \label{fig:QCSC_scheme}
\end{figure*}
We are now ready to present our main theoretical contributions.
We begin by introducing SqDRIFT, an algorithm that combines SKQD and qDRIFT into a hardware-friendly, sample-based and provably efficient protocol.
In SqDRIFT we compile the time-evolution circuit associated with the propagation time $t_k = kt$ by sampling $N_r$ qDRIFT randomized sequences, and collecting $S$ bitstrings from each realization.
The union of the overall $N_r \times S$ bitstrings defines the set of samples associated with time $t_k$.
The algorithm then proceeds as in conventional SKQD, i.e.\, the target Hamiltonian is diagonalized in the subspace defined by the union of the samples collected for all time-steps.
A scheme of the SqDRIFT pipeline is shown in \cref{fig:QCSC_scheme}.

Thanks to its close relationship with Krylov methods, we can immediately derive rigorous analytical guarantees for SqDRIFT.
Notably, these results extend the analogous ones holding for SKQD~\cite{Yu2025_SKQD}, incorporating the effect of unitary compilation errors~\cite{epperly2021subspacediagonalization,kirby2024analysis} and finite sampling statistics. 

\begin{theorem}[\textbf{SqDRIFT convergence guarantees}]\label{thm:sqdrift_main}
Let $H = \sum_i c_i h_i$ with $\lambda = \sum_i |c_i|$ be an $n$-qubit Hamiltonian whose ground
state $|\phi_0\rangle$ is $(\alpha^{(0)}_L, \beta^{(0)}_L)$-concentrated, and let $|\tilde{\phi}\rangle$ be the
lowest energy state supported on the $L$ important bitstrings in $|\phi_0\rangle$. Then if all the important $L$ bitstrings are successfully sampled in the SqDRIFT protocol, the resulting ground state energy estimate will have an error
\begin{equation}
    \langle \tilde{\phi}|H |\tilde{\phi}\rangle - \langle \phi_0|H |\phi_0\rangle \leq \sqrt{8}||H||\left(1-\sqrt{\alpha_L^{(0)}}\right)^{1/2}\,.
    \label{eq:sqdrift-energy-error}
\end{equation}
By using $N_r$ qDRIFT randomizations of length $N$, and by taking $S$ samples from each of these realizations, the failure probability (i.e.\ of failing to find all $L$ important bitstrings) is bounded by
\begin{equation}
    p_\text{fail}\leq L \left((1-\delta)\left(1-p\right)^S + \delta\right)^{N_r}
    \label{eq:fail-qdrift}
\end{equation}
for any $\delta>0$, where
\begin{align}
    p = \left(\frac{|\gamma_0|\sqrt{\beta_L}}{d} - \varepsilon\right)^2\,,
    \label{eq:fail-qdrift-p}
\end{align}
with
\begin{equation}
    \varepsilon = (d-1)t\lambda\left(\frac{(d-1)t\lambda}{N} + \sqrt{\frac{11\ln(2^{n+1}/\delta)}{N}}\right)\,,
    \label{eq:qdrift-epsilon}
\end{equation}
and
\begin{equation}
    \beta_L = \beta_L^{(0)} - 2\sqrt{2}\sqrt{1-\sqrt{1-\xi/\Delta}}\,.
\end{equation}
Here, $\Delta = E_1 - E_0$ is the energy gap between the ground and first excited states, $t=\pi / (E_{2^n - 1}-E_0)$ is the evolution time chosen to construct the Krylov vectors, see \cref{eq:krylov_vectors}, and $\xi$ is given by \cref{eq:xi-krylov}.
\end{theorem}

\begin{proof}
    A complete proof of \cref{thm:sqdrift_main} is given in \cref{appendix:proof_sqdrift_main}. Here we report for completeness a sketch of the main ideas. The first few steps, as well as the last one, closely follow the analogous proof given in Ref.~\cite{Yu2025_SKQD} for SKQD.\\

    \textbf{Step 1.} \cref{thm:KQD_QDRIFT}, proven in \cref{appendix:proof_thm_krylov}, implies that within the Krylov subspace constructed via qDRIFT we can find a state $|\psi\rangle$ approximating the exact ground state as
    \begin{equation}
        ||\ket{\psi} - \ket{\phi_0}||^2 \leq \tilde{\xi}= O\left(\frac{\xi}{\Delta}\right)\,.
    \end{equation}\\
    \textbf{Step 2.} If $\ket{\phi_0}$ exhibits $(\alpha_L^{(0)},\beta_L^{(0)})$-concentration, then
    $\ket{\psi}$ is $(\alpha_L,\beta_L)$-concentrated with
    \begin{align}
        \alpha_L &= \alpha_L^{(0)} -2\sqrt{\tilde{\xi}}&&\text{and}& \beta_L &= \beta_L^{(0)} -2\sqrt{\tilde{\xi}}\,.
    \end{align}\\
    \textbf{Step 3.} Defining $p^{(k)}(b_j)\coloneqq|\langle b_j|\psi^k\rangle|^2$ for each $k=0,\ldots,d-1$, we show that the concentration of the ground state implies that for each $1 \leq i\leq L$ there exists a Krylov state $k$ such that \begin{equation}\label{eq:pk_bj_main}
        |p^{(k)}(b_i)|\geq \frac{|\gamma_0|^2\beta_L}{d^2}\,.
    \end{equation}
    Therefore, in the remainder of the proof we usually omit the Krylov index $k$ and assume we are working with the state such that \cref{eq:pk_bj_main} holds.\\
    \textbf{Step 4.} Fix a bitstring $b_i$. We can then prove that, with probability at least $1-\delta$ over the random choice of a single qDRIFT sequence $\vecm=(m_1,\ldots,m_N)$, $|\sqrt{p(b_i)}-\sqrt{\tilde{p}_{\vecm}(b_i)}|\leq \varepsilon$, where $\varepsilon$ is given by \cref{eq:qdrift-epsilon} and $\tilde{p}_{\vecm}(b_i)$ is the (scalar) probability of measuring $b_i$ on the state prepared by the sampled unitary $V_{\vecm}$. In turn, this implies that with probability at least $1-\delta$ (again over the draw of $\vecm$), $|\sqrt{\tilde{p}_{\vecm}(b_i)}|\geq\frac{|\gamma_0|\sqrt{\beta_L}}{d} - \varepsilon$. Finally, we have that 
    \begin{equation}
    1-\tilde{p}_{\vecm}(b_i)\leq1-\left(\frac{|\gamma_0|\sqrt{\beta_L}}{d} - \varepsilon\right)^2\eqqcolon 1-p
    \end{equation} 
    with probability at least $1-\delta$, again for this single realization.\\
    \textbf{Step 5.} Consider one qDRIFT realization (one random draw of $\vecm$), from which we take $S$ samples. With probability at least $1-\delta$ over this draw, Step $4$ gives $\tilde{p}_{\vecm}(b_i)\geq p$, so the $S$ samples miss $b_i$ with probability at most $(1-p)^S$. In the remaining cases, which occur with probability at most $\delta$, we bound the miss probability trivially by $1$. Averaging over the draw of $\vecm$, a single realization thus misses $b_i$ with probability at most $(1-\delta)(1-p)^S+\delta$. Since the $N_r$ realizations are drawn independently, missing $b_i$ in \emph{all} of them (i.e.\ in all $N_r\times S$ shots) occurs with probability at most
    \begin{equation}
        p_{\text{fail}}(b_i) \leq \left((1-\delta)\left(1-p\right)^S + \delta\right)^{N_r}\,.
    \end{equation}
    In total, considering the probability of missing at least one of $L$ important bitstrings leads to \cref{eq:fail-qdrift}.\\
    \textbf{Step 6.} Similarly to Ref.~\cite{Yu2025_SKQD}, we conclude the proof by showing that the state $|\tilde{\phi}\rangle = (1/C)\sum_{j=0}^{L-1}a_j|b_j\rangle$ with $C = \sqrt{\sum_{j=0}^{L-1}|a_j|^2}$, representing the restriction on the $L$-dimensional basis of important bitstrings of the target ground state defined on the full $n$-qubit Hilbert space, represents a valid solution to the ground state approximation problem. Indeed, its energy is close to that of $|\phi_0\rangle$ as described in \cref{eq:sqdrift-energy-error}.
\end{proof}

\cref{thm:sqdrift_main} certifies that accurate approximations of the ground-state energy, which are provably achievable with KQD, are also efficiently attainable through the SqDRIFT sampling pipeline provided that the correct sparsity and initial overlap assumptions are met.
As for SKQD, $\alpha_L^{(0)}$ controls the approximation error, while $\beta_L^{(0)}$ affects the failure probability.
It is worth noticing that a similar analysis could be carried out for other types of unitary compilation errors, for instance replacing qDRIFT with Suzuki-Trotter product formulas, as well as for incoherent noise sources.
In the context of molecular systems, for which we usually have $\lambda \ll \mathcal{N} \max_i | h_i |$, the qDRIFT convergence is expected to be particularly favorable in practice.

We stress that, as for Krylov and quantum subspace diagonalization methods in general, these are worst-case convergence guarantees: they certify polynomial-time convergence under the stated assumptions but are not expected to be tight~\cite{epperly2021subspacediagonalization}. The same holds true for the underlying Hamiltonian-simulation bounds, which are known to overestimate empirical errors by orders of magnitude~\cite{Childs2018_FirstSimulation}. In practice, as we will show in \cref{sec:results}, SqDRIFT generates accurate subspaces well before these bounds are saturated. We also emphasize that, unlike generic Krylov guarantees on the eigenvalue error, \cref{eq:fail-qdrift,eq:fail-qdrift-p} bound the probability of sampling each of the important determinants, thereby directly characterizing the determinant-selection step that is central to selected-CI approaches.

Ideally, one would set the depth of each one of the $N_r$ randomized circuits (e.g.\, by selecting an appropriate number of terms $N$ in the qDRIFT sequences) based on some desired target accuracy in the approximation of exact time evolution.
In practice, however, we instead limit the maximum depth of each propagation circuit to a value that can be successfully implemented on current quantum processors.
In general, this choice will yield only a relatively crude approximation of the time-evolution channel.
Nevertheless, even under such conditions and in the presence of hardware noise, SqDRIFT remains fully compatible with the variational principle.
Therefore, its accuracy can be assessed by comparing \emph{a posteriori} the resulting energy to alternative methods, and a lower energy always signals a better result.

Recall also that the circuits used in any SQD flavor must produces samples in the subspace on which the ground-state wave function has a large overlap in order for the whole protocol to be efficient.
Hence, even a very crude approximation of the time-evolution operator may not severely impact the SqDRIFT simulation accuracy, provided that the approximation does not severely change the circuit support.

In the context of electronic structure problems, we will consider molecular Hamiltonians that, when projected onto a $N_\text{MO}$-dimensional molecular orbitals set, read
\begin{equation}
 \begin{aligned}
  H_\text{ele} &= 
    \sum_{pq}^{N_\text{MO}} \sum_{\sigma \in \left\{ \alpha, \beta \right\}}
        h_{pq} a_{p \sigma}^\dagger a_{q \sigma} \\
    &+ \frac{1}{2} \sum_{pqrs}^{N_\text{MO}}
      \sum_{\sigma, \sigma^\prime \in \left\{ \alpha, \beta \right\}}
        (pq | rs) a_{p \sigma}^\dagger a_{r \sigma^\prime}^\dagger a_{s \sigma^\prime} a_{q \sigma} \, ,
 \end{aligned}
 \label{eq:SQ_Electronic}
\end{equation}
where $h_{pq}$ and $(pq | rs )$ are one- and two-electron integrals, respectively.
A natural choice is then to map each individual string of fermionic operators, $a_{p \sigma}^\dagger a_{q \sigma}$ and $a_{p \sigma}^\dagger a_{r \sigma^\prime}^\dagger a_{s \sigma^\prime} a_{q \sigma}$, to a single $h_i$ operator among the ones appearing in \cref{eq:Hamiltonian_Decomposition}.
In this way, every term sampled in the qDRIFT construction corresponds to a full fermionic excitation term, which brings the key advantage that each individual propagator $\mathrm{e}^{-\mathrm{i} h_i t}$ preserves the Hamiltonian particle-number symmetry.
At the same time, however, implementing $\mathrm{e}^{-\mathrm{i} h_i t}$ for two-body excitations can be very demanding on limited-connectivity quantum processors, as these yield rather deep circuits using the Jordan-Wigner mapping scheme~\cite{Jordan1928_JW-Transformation}.
Below, we discuss a strategy to reduce the circuit depth.

\section{Computational Pipeline}

In this section, we outline the computational pipeline of the qDRIFT sampling protocol as well as the additional pre- and post-processing steps undertaken to make the circuits more amenable to execution on quantum computers and/or improve accuracy of the results.

\subsection{The sampling protocol}\label{sec:sampling_protocol}

As alluded to before, the qDRIFT protocol makes no assumption on the partition of the Hamiltonian (see \cref{eq:Hamiltonian_Decomposition}).
Applying the Jordan-Wigner mapping to any individual fermionic excitation term of the Hamiltonian given in \cref{eq:SQ_Electronic} will result in a sum of Pauli terms.
Therefore, the partition may be performed either before the mapping, on the fermionic terms and coefficients, or in the qubit basis, after the fermion-to-qubit mapping has been applied. %
In this work, we choose to sample Hamiltonian terms expressed in the fermionic basis because this allows us to perform two crucial optimizations: first, we can enforce particle-number conservation, a property that would get lost by sampling individual Pauli terms.
Second, we can optimize the fermion-to-qubit mapping procedure, to minimize the circuit depth of their subsequent implementation on quantum hardware, which we discuss in more detail in the next section.

Particle-number conservation is crucial for accurately sampling the relevant subspace and, hence, for the SQD algorithm to succeed. %
Furthermore, quantum-hardware implementations of SQD-based algorithms use the \emph{configuration recovery} method~\cite{RobledoMoreno2024_SQD-Original}, a heuristic technique that allows to deal with the effects of noise at the level of individual samples.
CR distinguishes correct and noisy sampled bitstrings based on whether they have the same particle number as the initial state.
Such an error-detection protocol can work only if symmetry breaking implies the presence of noise.
Therefore, we should strive towards enforcing particle-number conservation on the level of the quantum circuits employed for the sampling step.

The fermionic terms appearing in our Hamiltonian can be either 1- or 2-body excitation terms which
take the form $a^\dagger_{p\sigma} a_{q\sigma}$ and $a^\dagger_{p\sigma} a^\dagger_{r\sigma^\prime} a_{s\sigma^\prime} a_{q\sigma}$, respectively.
Mapping any one of these terms to the qubit basis results in a sum of up to $4$ or $8$ Pauli terms, respectively (since mapping any individual operator results in 2 Pauli terms).
This would imply a significant depth overhead when implementing the time-evolution operator on a quantum circuit level, as $4$ ($8$) sequential Pauli evolution operations would have to be performed on the same set of qubits.
However, we can leverage further symmetries of the fermionic Hamiltonian to mitigate this penalty.

Since fermions in a system are indistinguishable, symmetries arise that result in $1$- and $2$-body excitations on identical support (i.e.\ the sets $\{p, q\}$ and $\{p, r, s, q\}$, respectively) to share the same coefficients.
Therefore, the resulting Pauli terms after the fermion-to-qubit mapping will also share the same coefficients, albeit with different signs depending on the permutation of the indices.
Consequently, mapping the sum of all fermionic excitations that share the same index support together results in cancellation of terms such that the sum of resulting Pauli terms contains at most $2$ or $4$ terms in the case of $1$- and $2$-body excitations, respectively.
This effectively halves the circuit depth required to implement particle-number conserving excitations.

\subsection{The fermion-to-qubit mapping}\label{sec:F2Q}
\begin{figure}
  \centering
  \includegraphics{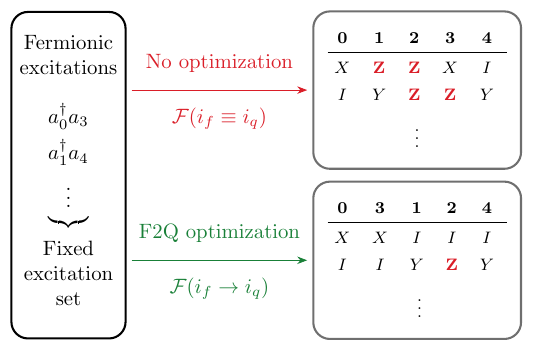}
  \caption{\textbf{F2Q layouting example}. Given a set of excitations, exemplified here by two single excitations, no optimization of the routing of fermionic spin orbitals to the qubit register leads to a naive layout with potentially high-weight Paulis. In contrast, an optimized layout leads to lower-weight Paulis (notice the re-shuffled indices), resulting in significant savings in terms of circuit depth. Note that, in principle, there is no reason why indices $1$ and $4$ should not be further placed next to each other. However, since the optimization must balance the weight of all excitations in the currently sampled batch, the optimal solution may not lead to minimal Pauli weight for \emph{every} single resulting Pauli string and, here, only an exemplary subset is shown for illustration purposes.}
  \label{fig:Layout}
\end{figure}
In the Jordan-Wigner mapping, a time-evolution term of the form $\mathrm{e}^{-\mathrm{i} a_{p,\sigma}^\dagger a_{q,\sigma} t}$ is mapped onto an exponential of highly non-local Pauli operators that include a chain of Pauli $Z$ operations acting on all qubits with index between $p$ and $q$.
The mapping is defined by a 1-to-1 mapping of each fermionic mode index $p_f$ ($q_f$) onto a qubit index $p_q$ ($q_q$).
The choice of the fermion-to-qubit (F2Q) index map, $\mathcal{F}(i_f \rightarrow i_q)$, can be optimized to maximize the locality of the sampled time-evolution terms.
This means that we minimize the distance between the qubit indices $p_q$ and $q_q$ for a specific fermionic excitation acting on $(p_f, q_f)$, in turn minimizing the Pauli weight of the resulting Pauli terms.
This process is depicted pictorially in \cref{fig:Layout}.
Such an optimization is possible only when working with a subset of sampled fermionic excitations rather than the full system Hamiltonian at any time.
In the limit in which all Hamiltonian terms are included (such as, e.g.\, in a Trotter-like decomposition), the reordering of $\mathcal{F}$ to minimize the distance of one particular excitation would necessarily imply the lengthening of another excitation.

Note that the function $\mathcal{F}$ underlying the JW mapping implicitly defines a 1-dimensional sorting of the qubits, independently of any hardware topology used later on.
Optimizing this sorting is therefore a one-dimensional problem, which can be solved straightforwardly.
The constraints should enforce that $\mathcal{F}$ is a valid permutation of indices, that is, every fermionic mode index must be mapped to exactly one unique qubit index.
Further, every fermionic excitation in the current set of sampled excitations will result in a penalty term that depends on its support.
Excitations which have support on only a single index (i.e.\ number operators) can be disregarded, because they will map to single-Pauli weight terms independently of $\mathcal{F}$.
In the case of $1$-body excitations, this leaves only terms of the form
$a_{p_f}^\dagger a_{q_f}$ where $p_f \neq q_f$.
For these terms, the goal is to minimize the distance of $|p_q - q_q|$, leading to Pauli terms of weight $2$ in the optimal case.

Minimizing the weight of two-body excitations is more difficult as their support ranges from $1$ to $4$.
A term with support $1$ is effectively a single-Pauli weight terms, which can be ignored as before.
The same holds for $2$-body excitations of support $2$, since this yields the form $a_{p_f}^\dagger a_{q_f}^\dagger a_{q_f} a_{p_f}$, i.e.\ two number
operators acting on $p_f$ and $q_f$, respectively.
While one \emph{could} minimize the distance of the corresponding qubit indices, this is not going to affect the weight of the mapped Pauli terms, which do not contain any chain of $Z$ operators connecting the two sites.
Therefore, the Pauli weight will be $2$, irrespective of $\mathcal{F}$.
It should be noted, that the layout of $p_q$ and $q_q$ on a constrained hardware topology \emph{will} affect the actual depth of the transpiled circuit, but this is a separate optimization problem to be solved later on.

That leaves us with $2$-body excitations of support $3$ and $4$.
The former case is equivalent to a $1$-body excitation of support $2$, with an additional number operator on some index $p_f$.
Just like in the case of support $2$, we can ignore the index of this number operator and postpone the layout optimization to the transpilation step.
Therefore, $2$-body excitations of support $3$ can be optimized by minimizing a single distance $|r_q - s_q|$.
Finally, $2$-body excitations with $4$ unique indices must minimize all involved index distances.
The weights of these terms are the hardest to minimize.

Finally, we note that indices absent from a given circuit leave the corresponding qubits idle; these could be removed to reduce the circuit width, which would be most beneficial in the low-excitation regime, where idle qubits are more frequent.

\subsection{Circuit synthesis}\label{sec:circuit_synthesis}

At this point each SqDRIFT randomization consists of a set of Pauli strings.
Synthesizing the quantum circuit that implements the corresponding time evolution of an initial state is carried out using Qiskit~\cite{Qiskit2024}.
In this work, we use Qiskit \texttt{v2.1.0} with its default transpiler pipeline set to the highest level of optimization for the targeted backend connectivity.
To further optimize the quantum circuits we enabled a high-level synthesis plugin leveraging \texttt{rustiq v0.0.8}~\cite{Debrugiere2024_rustiq} configured to preserve the order of Pauli operators to avoid biasing of the qDRIFT protocol. Aside from this constraint, we employ the standard default synthesis and transpilation pipeline, without further manual tuning of parameters. For further details on the software employed we refer to \cref{appendix:computational-details}.

\subsection{Extending the quantum subspace}\label{sec:ext-sqdrift}

Once samples have been collected from the quantum device and classical diagonalizations are complete, an additional -- optional -- step can be performed. 
Due to noise in current hardware, or partial convergence of the qDRIFT protocol, some important configurations arising from higher-order excitations may be lost.
To address this, we employ an extension of the SQD algorithm, indicated as Extended-SQD or Ext-SQD.
Originally introduced to improve the accuracy of SQD for excited state determination~\cite{Barison2025_SQD-ExcitedStates}, Ext-SQD was later adapted for precise calculation of thermodynamic uncertainty relations~\cite{motta2025_tur}.
In Ext-SQD, electronic configurations $b_i$ sampled from quantum circuits serve as the starting point to generate new configurations by applying to them low-order electronic excitation operators.
In particular, we define a set of operators 
$E \in \{ \mathbb{I} , a^{\dagger}_{a\sigma} a_{i\sigma} , a^{\dagger}_{a\sigma} a^{\dagger}_{b\tau} a_{j\tau} a_{i\sigma}, \dots \}$ 
and enlarge the original subspace through $\ket{\tilde{b}_i} = E \ket{b_i}$.

The resulting configurations $\tilde{b}_i$ may already exist in the original sample set -- in which case they are omitted -- or they may represent excited configurations that were missed. 
In this case, the inclusion of $\tilde{b}_i$ in the quantum subspace improves the accuracy of the final result.
We denote the subspace obtained by extending the SqDRIFT samples as Ext-SqDRIFT ($n$), where $n = \mathrm{S}, \mathrm{SD}, \mathrm{SDT}, \ldots$ indicates the rank of electronic excitation operators applied to the original wave function.
For a full description of the method, including its computational cost and limitations, we refer the reader to~\cite{Barison2025_SQD-ExcitedStates}.

\section{Results}\label{sec:results}

To showcase SqDRIFT, we consider two systems from the family of the polycyclic aromatic hydrocarbons (PAHs).
PAHs are of key interest due to their relevance in environmental science (pollutants, carcinogens) and materials science (organic semiconductors).
Moreover, their extended $\pi$-conjugated systems lead to unique electronic, optical, and magnetic properties of physical interest~\cite{Harvey1997_PAH-Book,Lawal2017_PAH}.
The size of these systems (and, consequently, the complexity of the corresponding electronic-structure simulations) can be systematically increased. 
For this reason, PAHs are commonly used as a test set to benchmark new electronic-structure calculation methods~\cite{Hachmann2007_PAH-DMRG,Hajgat2009_PAH-Benchmark,Schriber2018_PAH-ACI,Sharma2019_PAH-DMRG-DFT}.
In our experiments we first study naphthalene, using an active space including all the $\pi$-type orbitals, resulting in 10 electrons and 10 orbitals.
This system is small enough that the ground-state energy can be computed exactly with FCI and that the corresponding SqDRIFT circuits can be simulated exactly.
For hardware experiments instead, we switched to coronene, again including all $\pi$ orbitals, which yields an active space with 24 electrons and 24 orbitals.
At this scale, only approximate classical solutions can be obtained on non-specialized computational infrastructure.

Unless stated otherwise, all numerical and hardware experiments combine the samples obtained from groups of SqDRIFT circuit randomizations using three different $k$ values $\{1,2,3\}$ to scale the evolution time (i.e.\ the Krylov order index, cf. \cref{eq:krylov_vectors}).
In all simulations, the time-step $t$ is set to 1 (in inverse Hamiltonian units).
SqDRIFT simulations executed with other $t$ values yield qualitatively equivalent results to the ones reported in the following.
The overall number of circuit randomizations is $500$ per group (totaling $1500$ circuits) in the numerical simulations, and $1000$ randomizations per group (totaling $3000$ circuits) for the simulations executed on quantum hardware.
In the numerical section we simulate $512$ noiseless samples for each circuit while we collected $1024$ shots from each circuit for the hardware experiments.
\begin{figure*}
\subfloat[\label{fig:ConvergenceExcitations}]{%
\includegraphics{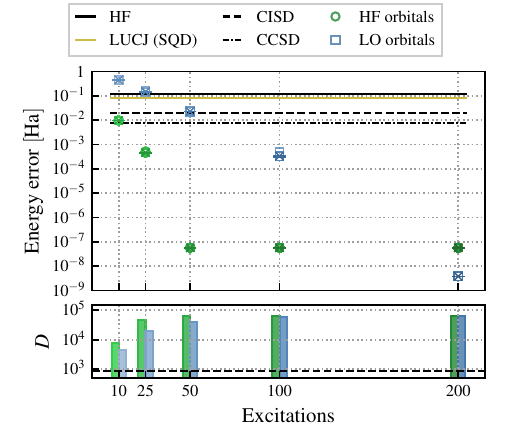}
}
\subfloat[\label{fig:ConvergenceRandomizations}]{%
\includegraphics{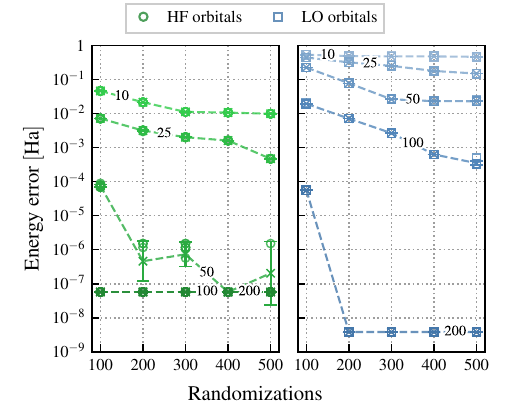}
}
\caption{\textbf{Convergence of the energy error of naphthalene with the length and number of SqDRIFT sequences (simulation).}
In both panels we plot the energy error of noiseless SqDRIFT simulations with respect to the exact energy obtained by diagonalization of the full system Hamiltonian.
The hollow circles in both panels correspond to repeated SqDRIFT simulations where each one is given a randomly selected $90\%$ of the available unique samples obtained from the noiseless quantum circuit simulations.
(a) Along the $x$ axis we vary the number of excitations sampled from the Hamiltonian (i.e.\ the SqDRIFT length) for each randomized circuit.
The green and blue circles indicate the Hartree-Fock (HF) and localized (LO) orbitals, respectively.
The HF and Configuration-Interaction including Single and Double Excitations (CISD) reference values are indicated by the solid and dashed black line, respectively.
The yellow line corresponds to the SQD energy obtained from the noiseless samples collected from the LUCJ ansatz with its parameters fixed to the $t_1$ and $t_2$ amplitudes of a prior CCSD calculation.
In the bottom panel, we show for both types of orbitals and for each energy value the subspace dimension used in the diagonalization. The dashed black line marks the CISD subspace dimension.
(b) Here the $x$ axis indicates the increasing number of SqDRIFT randomizations from which the samples are included in the pool of available bitstrings to sample from.
The green and blue lines indicate the convergence of the energy error within the HF and LO orbitals, respectively.
The shade of color indicates the number of excitations included in each randomized circuit (as indicated in subfigure a) as well as the overlaid number.
In the case of the HF orbitals, the lines of $100$ and $200$ excitations are superimposed.
}
\label{fig:ConvergenceQDriftParams}
\end{figure*}
\subsection{Classical simulations}
For an initial analysis of the SqDRIFT protocol we rely on noiseless classical simulations on naphthalene, using the active space defined above. %
For reference, the Hartree-Fock overlap of the exact naphthalene ground state is $|\gamma_0|^2 = 0.81$, confirming its single-reference-dominated, strongly concentrated character in the canonical basis.

Due to the structure of the system and the nature of the sampling protocol of SqDRIFT, we compare two orbital basis sets: localized orbitals (LO) and canonical Hartree-Fock (HF) orbitals.
The localized orbitals are obtained with the Boys localization scheme as implemented in PySCF~\cite{Sun2020_PySCF}.
We expect LO orbitals to result in shallower SqDRIFT circuits compared to the HF ones.
In fact, orbital localization will reduce the extent of long-range interaction terms, which will therefore be sampled with a lower probability.
We give in \cref{fig:F2QOpt} an example of the circuit depth reduction that can be obtained due to the F2Q layout optimization described earlier.
Note that the depth reduction due to the F2Q layout optimization diminishes by increasing number of excitations included in each SqDRIFT circuit.
This is to be expected since the larger number of excitations are more likely to result in conflicting constraints, such that optimizing the interaction distance of any one excitation is more likely to penalize another one.
In the remainder of this section we analyze the convergence behavior of SqDRIFT in the two chosen orbital bases: HF and LO.

We analyze the convergence with respect to three parameters.
The first two directly impact the quality of the solutions produced by SqDRIFT as per \cref{thm:KQD_QDRIFT} in \cref{appendix:proof_thm_krylov}, namely (i) the number of excitations sampled from the Hamiltonian (or length of qDRIFT sequences; $N$ in \cref{eq:qDRIFT_Unitary}) and (ii) the number of randomizations (or independent realizations of qDRIFT; $N_r$ in \cref{thm:sqdrift_main}).
The third parameter, the size of the diagonalization subspace, is the key parameter of the SQD post-processing step.
In the following we test their individual influence separately.

In \cref{fig:ConvergenceQDriftParams}, we show the scaling of the energy error as a function of the qDRIFT parameters.
Both subfigures indicate that the SqDRIFT energy can be converged to numerical precision for a fixed subspace dimension and irrespective of the chosen orbital basis.
This indicates that the chosen diagonalization subspace is sufficiently large to cover the support of the ground-state wave function.
Notably, the convergence with respect to the number of excitations is much slower for the LO orbitals than for HF (\cref{fig:ConvergenceExcitations}).
This result is expected because the overlap between the starting determinant and the exact wave function is smaller for LO than for HF ones.
Therefore, as expected by the bounds reported above, SqDRIFT requires deeper circuits. \cref{fig:ConvergenceExcitations} also shows that the samples obtained from a noiseless simulation of the LUCJ ansatz provide only a modest improvement over the initial HF energy estimate, whereas SqDRIFT circuits containing as few as 10 fermionic excitations already yield energies below the CISD reference in the HF orbital basis. We stress that the number of excitations characterizes the individual randomized circuit, and hence the quantum resources required, and should not be conflated with the dimension of the diagonalization subspace, which is an independent simulation parameter bounded by the total number of distinct configurations collected across all randomizations. %
\cref{fig:ConvergenceRandomizations} analyses the SqDRIFT convergence with respect to the number of qDRIFT samples.
The energy consistently improves by increasing the number of samples (with the only exception of the curve obtained for $N_r=500$ and $N=50$, which is discussed more in detail below).
This indicates that SqDRIFT allows for circuit depth to be traded in for randomized circuit realizations: fewer randomizations are required when sufficiently deep circuits are being sampled from, and vice versa.
The convergence is faster when increasing the circuit depth, as can be seen by the faster descent of the individual curves along the $y$ axis compared to the slower slope of any curve along the $x$ axis.
This is consistent with known concentration effects in qDRIFT: for sufficiently deep circuits, even a single randomization yields an accurate representation of the time-evolution operator. %
Lastly, we comment on the dispersion of the energy estimates across independent repetitions, which is largest for $50$ excitations in the HF orbital basis. As discussed in \cref{appendix:sampling-behavior}, for a small number of excitations the circuits sample nearly uniformly over a limited set of reachable determinants, so that independent repetitions collect almost the same configurations and yield highly reproducible -- but comparatively inaccurate -- energies. As the number of excitations grows, the sampled distribution broadens and extends into the low-amplitude tail of the wave function, so that the composition of the subspace varies more between repetitions, while the energy error decreases by orders of magnitude. The dispersion is further amplified by the logarithmic scale: in the regime of very small errors, a repetition that misses one of the smallest-amplitude determinants incurs a large relative but small absolute error, so that the distribution of errors is strongly right-skewed.
\begin{figure*}
\subfloat[Simulated results\label{fig:ConvergenceSubspace}]{%
\includegraphics{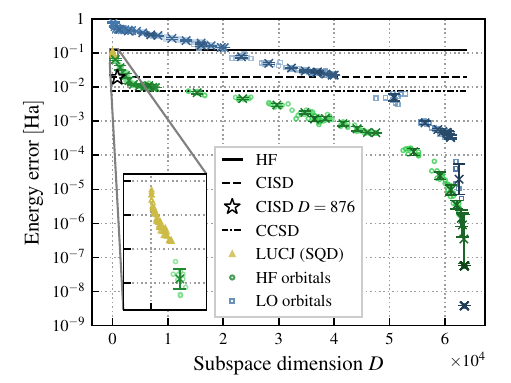}
}
\subfloat[Hardware experiments\label{fig:Naphthalene}]{%
\includegraphics{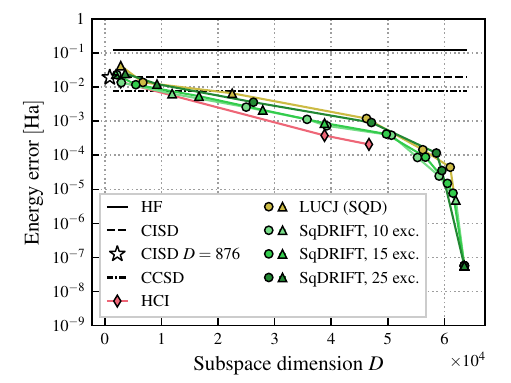}
}
\caption{\textbf{Energy error of naphthalene as a function of the diagonalization subspace dimension.} Analogously to \cref{fig:ConvergenceQDriftParams} we plot the energy error with respect to the exact energy obtained with FCI.
Here $D=\mathrm{d}_\alpha\mathrm{d}_\beta$ denotes the dimension of the diagonalization subspace, with $\mathrm{d}_{\alpha/\beta}$ the number of unique configurations retained per spin sector.
(a) Numerical simulation: for the results obtained from noiseless circuit simulations, the distributed data is obtained by taking increasing percentages from $\{10, 20,\dots, 100\%\}$ of the number of unique samples collected.
Each such simulation is repeated $10$ times (indicated by the hollow circles) for the data obtained from different SqDRIFT circuit depths (as indicated by the color shade akin to \cref{fig:ConvergenceQDriftParams}).
Since the final subspace dimension cannot be controlled directly, due to the classical post-processing stage, the different simulations place slightly differently along the $x$ axis.
A cross with error bars is placed at the center of each set of repeated simulations.
(b) Hardware data: the lowest energy estimates obtained throughout $3$ iterations of SQD post-processing and CR applied to the samples obtained from hardware experiments.
Circles indicate the lowest energy estimate obtained without the use of CR while triangles indicate that CR resulted in an energy decrease. The following classical computations serve as references: Hartree-Fock (black solid; $-383.384~\mathrm{Ha}$), CISD (black dashed; $-383.487~\mathrm{Ha}$), CCSD (black dash-dotted; $-383.498~\mathrm{Ha}$), and HCI (red; panel (b) only). In addition, the star marks the CISD energy placed at the dimension of the CISD space for this active space ($D=876$ determinants: the reference determinant, $50$ single and $825$ double excitations), allowing the different methods to be compared at equal diagonalization cost.}
\end{figure*}
We complete our numerical analysis with \cref{fig:ConvergenceSubspace} which plots the convergence of the energy error with respect to the dimension of the diagonalization subspace.
The results shown here are obtained from SQD post-processing steps given subsets of samples ranging from $10\%$ to $100\%$ of all collected bitstrings.
As explained earlier, a total of $\num{768000}$ samples are generated from noiseless circuit simulations ($512$ samples from $500$ randomized circuits at $3$ different $k$ values).
However, the right limit of \cref{fig:ConvergenceSubspace} is reached much earlier because only unique bitstrings will be included in the diagonalization subspace.
Furthermore, the size of the FCI Hilbert space is $\num{63504}$, indicating the upper bound of the $x$ axis for naphthalene.
The energy errors reported in \cref{fig:ConvergenceExcitations} provide a lower bound along the $y$ axis of \cref{fig:ConvergenceSubspace} for the results obtained from the corresponding circuit depths.
As expected, the best energy estimates are obtained with the largest subspace dimension. %
Importantly, \cref{fig:ConvergenceSubspace} highlights a key difference between the LO and HF orbital results.
While the latter surpasses the CISD reference value already at low subspace dimensions and reaches an accuracy of $\approx 1~\mathrm{mHa}$ for a subspace of $D\approx\num{40000}$, the LO orbital results require subspace dimensions that almost saturate the full Hilbert space to reach the same accuracy.
This indicates that the choice of the orbital basis plays a crucial role in SqDRIFT: on the one hand, it determines the sparsity of the Hamiltonian -- and hence the complexity of the qDRIFT circuits -- but on the other hand it critically tunes the concentration of the wave function.%

It is interesting to observe that the energy obtained from the noiseless LUCJ ansatz does not vary significantly with increasing subspace dimension.
This indicates that, in the setting studied here -- where the LUCJ parameters are initialized from CCSD amplitudes and are not subsequently re-optimized -- only a small number of unique bitstrings are sampled by that ansatz. This behavior is consistent with the limitations reported for sample-based approaches that draw from strongly peaked ground-state distributions~\cite{reinholdt2025critical}, and with recent analyses showing that an LUCJ circuit initialized from classical amplitudes, without subsequent quantum optimization, samples a compact and biased distribution peaked around the reference determinant~\cite{vaquerosabater2026noise}. This behavior is caused by both the truncation of the ansatz repetitions and the removal of interactions imposed by the device connectivity, an effect that can be reduced by improved initialization strategies~\cite{lin2025lucjinit} or by optimizing the parameters in a closed loop~\cite{shirakawa2025closedloop}. In contrast, and consistently with the coverage guarantee established by our convergence analysis, the SqDRIFT circuits sample a substantially less peaked distribution, which nonetheless retains its weight on the configurations on which the ground state has large support (\cref{fig:FCIvsSqDRIFT}); this is closely related to the mechanism by which time-evolution-based QSCI variants reach compact subspaces, namely by combining stochastic evolution with an occupancy-guided expansion of the sampled configurations~\cite{weaving2025compact}.

Lastly, to further characterize the sampling behavior of SqDRIFT, and based on the study reported in Ref.~\cite{reinholdt2025critical}, we compare in \cref{fig:FCIvsSqDRIFT} in \cref{appendix:sampling-behavior} the energy error of SqDRIFT with that obtained by sampling determinants directly from the exact full-CI ground-state wave function. Because the ground-state distribution is strongly peaked, sampling from it is inefficient at discovering the rare but relevant determinants, which is a limitation of ground-state-sampling QSCI recently emphasized in Ref.~\cite{reinholdt2025critical}. In contrast, at a fixed sampling budget SqDRIFT with a sufficient number of excitations reaches energy errors below chemical accuracy and several orders of magnitude smaller than exact ground-state sampling, since the qDRIFT circuits generate a less peaked distribution that progressively re-weights toward the true determinant-importance ordering (\cref{fig:SampledDistributions}). This ultimately confirms that circuit approximating the ground-state wave function of a given Hamiltonian are not the best choice to be used in SQD-based approaches.

From this numerical analysis we also conclude that, despite the significant reduction in quantum circuit depth, the use of the LO compared to the HF orbitals results in poorer convergence of SqDRIFT overall.
This is because the wave function is less concentrated in the LO orbitals than in the HF ones, rendering the classical post-processing step of SQD harder.
Therefore, we use the HF orbital basis in the remainder of this work.
\begin{figure*}
\subfloat[\label{fig:coronene}]{%
\includegraphics{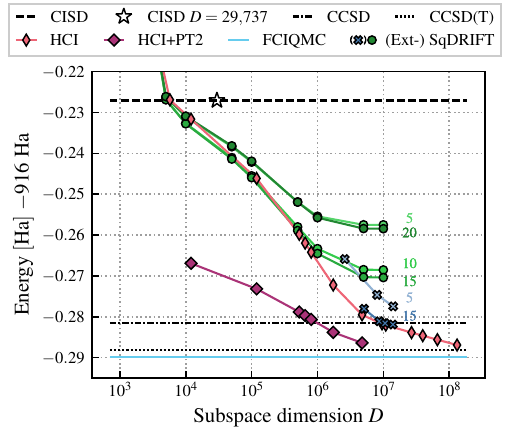}
}
  \subfloat[\label{fig:convergence-fciqmc}]{%
\includegraphics{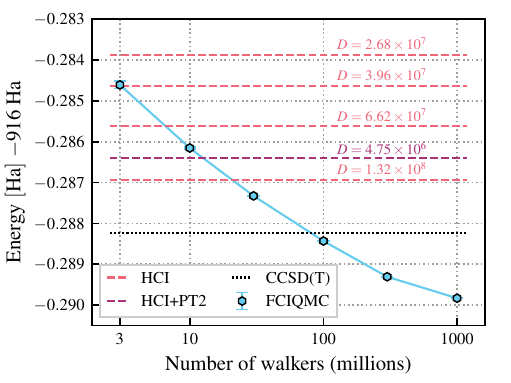}
}
  \caption{\textbf{Energy of coronene as a function of the diagonalization subspace dimension and convergence of the FCIQMC reference calculation (hardware)}.
  (a) Absolute ground-state energy of coronene (axis offset by $-916~\mathrm{Ha}$) as a function of the diagonalization subspace dimension $D$. Each SqDRIFT curve (green, labeled by the number of excitations per randomized circuit, from $5$ to $20$) is obtained from a \emph{single} diagonalization in the full sampled subspace: the resulting configuration-interaction wave function is ordered by decreasing magnitude of its CI coefficients and split into nested subspaces of increasing dimension $D$, from which the reported energies are read off.
  The classical references are CISD (black dashed; $-916.227~\mathrm{Ha}$), CCSD (black dash-dotted; $-916.282~\mathrm{Ha}$), CCSD(T) (black dotted; $-916.288~\mathrm{Ha}$) and FCIQMC (cyan; $-916.290~\mathrm{Ha}$, our best estimate of the exact energy); heat-bath CI (HCI, red) and HCI with a deterministic PT2 correction (purple) are shown as a function of their own subspace dimension.
  The star marks the CISD energy placed at the dimension of the CISD space ($D=\num{29737}$ determinants), so that all methods can be compared at equal diagonalization cost.
  The SqDRIFT results saturate in two respects: with respect to the number of excitations -- beyond which the increased circuit depth degrades the estimate, the $20$-excitation data being dominated by hardware noise -- and with respect to the subspace dimension, beyond which increasing $D$ no longer lowers the energy. Blue crosses show the Ext-SqDRIFT extension of the $5$- and $15$-excitation subspaces, obtained by applying low-order excitation operators to the sampled configurations.
  (b) Convergence of the FCIQMC reference with the number of walkers; the HCI (and HCI+PT2) solutions are drawn as horizontal lines, with the size of the corresponding subspace indicated above each line.
  }
  \label{fig:coronene-main}
\end{figure*}
\subsection{Hardware experiments}

In this section we execute the SqDRIFT workflow using the IBM Quantum Heron r3 processor, \texttt{ibm\_boston}.
First, for naphthalene, we compare the results to the ones obtained in the classical simulations, from the previous section.
Then, we target coronene to benchmark SqDRIFT on a system for which brute-force exact simulations become unfeasible.

\subsubsection{Naphthalene}

\cref{fig:ConvergenceExcitations} shows that the HF orbital-based SqDRIFT yields an energy lower than the CISD reference value by including $10$ excitations.
Moreover, including $25$ excitations yields chemical accuracy.
We verify if the same trend is observed on hardware experiments by executing SqDRIFT using $10$ and $25$, as well as an intermediate value of $15$, excitations.
For $k=1$ these circuits reach $69\pm69$ ($113\pm116$), $140\pm117$ ($246\pm217$), and $355\pm222$ ($692\pm467$) two-qubit gate depths (counts) for $10$, $15$, and $25$ excitations, respectively.
Since the circuit shape is independent of $k$ the statistics are similar for other $k$ values.
More details are listed in \cref{tab:hw_nap}.
For comparison, the LUCJ ansatz yields a two-qubit gate circuit depth of $50$ and count of $324$.

\cref{fig:Naphthalene} reports the energy error obtained from the hardware experiments.
The noiseless simulation results and hardware experiments follow the same trend.
Moreover, for this specific test-case, we observe that the combination of hardware noise and CR error mitigation facilitates the exploration of the relevant subspace. 
This is best seen in the case of LUCJ for which the noiseless simulation yielded a small improvement over HF due to the small number of sampled determinants while, for hardware experiments, larger subspaces are reached due to the CR-based classical post-processing step.
Although similar conclusions can be drawn for SqDRIFT, this behavior is not expected to hold as the system size is increased, due to the exponentially-large number of irrelevant determinants that the effect of noise may yield.

\subsubsection{Coronene}

Now we turn our attention to the results obtained for coronene which are summarized in \cref{fig:coronene-main}.
With $48$ qubits, this system approaches the limits of the brute-force method for exact diagonalization; therefore, the energy of its ground state must be approximated.
Here, we report results obtained from four different classical methods: HF, CISD, FCIQMC, and \emph{heat bath configuration interaction} (HCI)~\cite{Holmes2016_HBCI,holmes2016efficient,Sharma2017_Semistochastic}.
HCI calculations were executing by using different values of the truncation threshold $\varepsilon_1$ are considered between $\varepsilon_1 = 10^{-3}$ and $\varepsilon_1 = 3 \cdot 10^{-6}$, yielding increasingly large subspace dimensions.

The classical reference ground-state energy of coronene is obtained using FCIQMC~\cite{Booth2009_FCIQMC}.
Specifically, we apply FCIQMC within the adaptive-shift approximation~\cite{Ghanem2019_InitiatorUnbiasing-FCIQMC,Ghanem2020_AdaptiveShift-FCIQMC} using up to $10^9$ walkers.
The convergence of this reference calculation with the number of walkers is reported in \cref{fig:convergence-fciqmc}.

The coronene hardware experiments are characterized by four specific features.
First, all randomizations of a given experiment are submitted within a single Qiskit Runtime session, so that they share a common device-calibration context.
Second, the physical qubit register is selected from an in-situ readout characterization rather than from the reported calibration data: prior to each job we run a readout-twirled measurement program ($300$ twirling randomizations of $25$ shots each) over the full device and form the twirl-corrected outcomes by taking the exclusive-OR of the raw measurements with the recorded twirl flips, which yields for every qubit a readout error $e_q$ given by the mean of $p(1|0)$ and $p(0|1)$.
Qubits with $e_q > 0.20$ are removed from the coupling map, all simple paths of exactly $2N_\mathrm{orb}$ qubits are enumerated on the pruned graph, and the candidate chains are ranked by their joint readout fidelity $\prod_q (1-e_q)$; the highest-ranked chain defines the initial layout and the pruned coupling map constrains the routing.
This device-level selection composes with the fermion-to-qubit assignment of \cref{sec:F2Q}, in which a mixed-integer program (Pyomo/HiGHS) reorders the spin orbitals independently for each randomization so as to minimize the Jordan--Wigner string weight of the sampled excitations, following the procedure described above; the resulting permutation is stored with each circuit and inverted when the measured bitstrings are merged.
Third, the qDRIFT sampling is restricted to \emph{effective} excitations: in the standard scheme the $N$ Pauli-rotation slots of a randomization are drawn independently with probability $p_j = c_j/\lambda$, with $\lambda = \sum_i |c_i|$ as in \cref{eq:qDRIFT_Unitary}, corresponding to a qDRIFT step size $\lambda t/N$; a substantial fraction of the drawn excitations, however, may act trivially on the state (e.g.\ core-core excitations), so that the corresponding rotation reduces to the identity and the slot is wasted.
We therefore rejection-sample each slot against a pair of index sets, initialized to the occupied (annihilable) and empty (creatable) spin orbitals of the initial determinant and grown as sampling proceeds: an excitation is accepted only if each of its particle--hole rails -- one for a single excitation, two independent rails for a double -- has one endpoint in the creatable set and the other in the annihilable set, and upon acceptance both touched spin orbitals are added to both sets, reflecting that the corresponding rotation may or may not have fired.
Since the rejected terms act as the identity on the state, discarding and redrawing them leaves the realized state unchanged: the procedure is a state-aware counterpart of the compiler pass that removes identity-equivalent operations, and it guarantees that the full budget of $N$ rotations is spent on terms acting non-trivially.
Notably, due to this additional step, a given number of \emph{effective} excitations will correspond to a smaller number of non-transpiled fermionic excitations (this difference must be taken into account when comparing the results obtained for naphthalene and for coronene).
Fourth, the classical post-processing employs a carry-over mechanism within the self-consistent configuration-recovery loop~\cite{shirakawa2025closedloop}: after each recovery iteration, every determinant of the lowest-energy batch whose CI amplitude exceeds $10^{-4}$ has its $\alpha$ and $\beta$ strings extracted, ordered by decreasing marginal weight $\sum |c|^2$, and prepended to the subspaces constructed in the following iteration before truncation at the maximum subspace dimension.
High-weight determinants are therefore kept across iterations independently of the stochastic re-subsampling, which stabilizes the recovery loop and makes its energy monotonically decreasing.
The diagonalizations use $10$ batches of $5\times10^3$--$6\times10^3$ samples, up to $5$ recovery iterations and spin-symmetrized $\alpha/\beta$ subspaces, and are solved with the distributed \texttt{SBD} selected-CI solver.

The resulting SqDRIFT energies are reported in \cref{fig:coronene}; their behavior with respect to the number of excitations and to the subspace dimension is described in the caption. In particular, already at a subspace dimension smaller than that of the CISD space ($D=\num{29737}$), SqDRIFT drops below the CISD reference, and it then follows the HCI curve over a range of $D$ before saturating; the lowest SqDRIFT energies lie well past CISD and approach the CCSD estimate. Consistently with the SqDRIFT error bounds, the energy improves as more excitations are included in the circuits, up to the point where hardware noise dominates and reverses the trend, most visibly at $20$ excitations. This saturation is overcome by the Ext-SqDRIFT extension~\cite{Barison2025_SQD-ExcitedStates}: augmenting the sampled subspace with the configurations generated by applying low-order excitation operators to the SqDRIFT determinants (blue crosses in \cref{fig:coronene}) lowers the energy well below the plain-SqDRIFT plateau, with the $15$-excitation extension reaching the CCSD reference at the cost of a single post-processing step. Hardware results obtained for the same systems on a IBM Quantum Heron r2 processor are reported for completeness in \cref{appendix:additional-data}, which additionally includes the comparison with the LUCJ ansatz.

\section{Conclusions and outlook}

We have introduced SqDRIFT, a new algorithm that bridges the Krylov sample-based approach introduced in~\cite{Yu2025_SKQD} with the qDRIFT randomized compilation strategy for time evolutions~\cite{Campbell2019}.
By reducing the depth of the time-evolution circuits, SqDRIFT extends the application range of SKQD to fermionic many-body Hamiltonians with complex interaction terms, such as molecular electronic-structure ones.
We have proved that SqDRIFT preserves the same convergence guarantees of SKQD.
This implies that the ensemble of randomized time-evolution circuits generated by qDRIFT constitutes a set of good sampling states for a subspace method.
Moreover, classical SqDRIFT simulations of naphthalene have shown that qDRIFT can generate, in practice, accurate subspaces before the standard qDRIFT convergence bounds are saturated.
Quantum computations based on SqDRIFT for coronene illustrate that SqDRIFT can provide reliable energy estimates even when applied to systems that are at the edge of exact classical simulations.
In essence, our contribution is a constructive recipe for a better sampling subspace and, in turn, a better energy estimate: one improves the subspace within a fixed reference state by increasing $N$, $N_r$, $S$, and $d$.\\

\textit{Extensions.} SqDRIFT can be combined with all the extensions to the SQD method that have been recently put forward, including excited-state calculations~\cite{Barison2025_SQD-ExcitedStates}, inclusion of dynamical correlation effects~\cite{Danilov2025_SQD-AFQMC}, the combination with embedding schemes~\cite{Shajan2025_SQD-Embedding,Kaliakin2025_SQD-Solvent}, and entanglement forging~\cite{smith2025quantumcentricsimulationhydrogenabstraction}.
The generic SqDRIFT-based theoretical framework derived here can be further optimized to facilitate its execution on current quantum computers.

\textit{Orbital basis.} Choosing an appropriate orbital basis crucially impacts the efficiency of SqDRIFT.
If, on the one hand, canonical Hartree-Fock orbitals yield a relatively fast convergence of SqDRIFT with the subspace dimension, on the other hand they produce deep time-evolution circuits, even after the qDRIFT compilation.
Localized orbitals yield much shallower circuits, since long-range interaction terms are not sampled.
The price to pay is, however, that the ground-state wave function becomes less concentrated.
Other choices of the orbital basis (e.g.\, natural or split-localized orbitals) may provide a better tradeoff between circuit complexity and SqDRIFT efficiency.

\textit{Fermion-to-qubit mapping.} Independently of the choice of the orbital basis, the circuit depth of the SqDRIFT circuits crucially depends on the adopted fermion-to-qubit mapping.
In the present work, we have shown that a sample-specific optimization of the Jordan-Wigner mapping~\cite{Jordan1928_JW-Transformation} can greatly reduce the circuit depth.
However, the optimization becomes increasingly less beneficial for deep circuits, for which the non-locality of the Jordan-Wigner transformation cannot be circumvented.
For such circuits, applying recently-proposed compact fermion-to-qubit mappings~\cite{Derby2021_FermionMapping,Algaba2025_F2Q-Compact,Nigmatullin2025_F2Q-Crossover} may drastically reduce the depth of a single Trotter step, and, therefore, enable realizing deeper qDRIFT circuits.

\textit{Hardware-friendly compilation.} The SqDRIFT circuits could be further rendered more hardware friendly by biasing the selection of the Hamiltonian terms towards hardware-friendly Trotter steps.
Although this bias affects the asymptotic error of SqDRIFT, it can be controlled as described in Ref.~\cite{Kiss2023}.
Moreover, we expect that the bias effect will be balanced by the smaller impact of the noise.

\textit{Initial state and iterative refinement.} As for KQD and SKQD, the efficiency of SqDRIFT also depends on the choice of the initial state for the quantum simulation.
When the ground-state weight is spread over an exponentially growing number of configurations -- the genuine strong-correlation limit -- no determinant-selection scheme, sampling from a single reference or otherwise, can capture the support with polynomial resources; this is a property of the target state, not of the initial configuration.
We expect instead that systems whose ground state remains concentrated but is poorly described by a single reference can be efficiently simulated with SqDRIFT through an iterative procedure.
A first SqDRIFT calculation employing, e.g.\, the Hartree-Fock determinant as the initial state can be used to identify the predominant configurations.
These configurations can be used as the starting point for a second SqDRIFT calculation, and the procedure can be repeated until self-consistency.
Such a procedure mimics the iterative subspace construction of classical selected CI methods~\cite{Sharma2017_Semistochastic}, and may enable a more efficient exploration of the relevant subspace.

To facilitate the adoption of the method, an efficient open-source implementation of SqDRIFT is available in the \texttt{qiskit-fermions} library~\cite{qiskit-fermions}.

\begin{acknowledgments}
We thank Stefan W\"orner, Simon Martiel, Julien Gacon, Elena Pe\~na Tapia, and Laurin Fischer for fruitful discussions. This research was supported by NCCR MARVEL, a National Center of Competence in Research, funded by the Swiss National Science Foundation (grant number 205602) and by RESQUE funded by the Swiss National Science Foundation (grant number 225229). This work was supported by the Hartree National Center for Digital Innovation, a UK Government-funded collaboration between STFC and IBM.
\end{acknowledgments}
\bibliography{references}

\clearpage
\onecolumngrid
\appendix
\section{Technical proofs}
Throughout this Appendix we adopt the following notation.
\begin{description}
    \item[$U^j \coloneqq \mathrm{e}^{\mathrm{i}jHt}$] The ideal evolution on $n$ qubits with ideal state $\ket{\psi^j}\coloneqq U^j\ket{\psi_0}$.
    \item [$V^j_{\vecm}$] A sampled unitary operator with indices $\vecm=(m_1,\ldots,m_N)$ and $\ket{\tilde{\psi}^j_{\vecm}}\coloneqq V^j_{\vecm}\ket{\psi_0}$.
    \item [$p^j_{\vecm}$] The probability of sampling $V^j_{\vecm}$.
    \item [$V^j\coloneqq \frac{1}{N_r}\sum_{r=1}^{N_r}V^j_{\vecm_r}$] The implemented qDRIFT protocol with $\ket{\tilde{\psi}^j}\coloneqq V^j\ket{\psi_0}$.
    \item [{\(\mathbb{E}[V^j]\coloneqq \sum_{\vecm} p^j_{\vecm} V^j_{\vecm}\)}] The ideal qDRIFT protocol, i.e.\ $V^j$ when $N_r\rightarrow\infty$, and $\mathbb{E}[\ket{\tilde{\psi}^j}]\coloneqq \mathbb{E}[V^j]\ket{\psi_0}$.
\end{description}

\subsection{Krylov quantum diagonalization with qDRIFT compilation}\label{appendix:proof_thm_krylov}

\begin{lemma}[\textbf{Krylov quantum diagonalization with qDRIFT}]\label{thm:KQD_QDRIFT}
Let $H = \sum_i c_i h_i$ with $\lambda = \sum_i |c_i|$ be an $n$-qubit Hamiltonian with eigenvalues $E_0\leq E_1\leq\ldots$. Consider the Krylov subspace spanned by
\begin{equation}
    |\psi_k\rangle = \mathrm{e}^{-\mathrm{i}kHt}|\psi_0\rangle\,,
\end{equation}
for $k\in\{0,1,\ldots,d-1\}$ (with $d$ odd) and some initial reference wavefunction $|\psi_0\rangle$. If all states $\{|\psi_k\rangle\}$ are prepared using $N_r$ qDRIFT randomizations of length $N$, then, for any $\delta > 0$, the approximate ground state energy $\tilde{E}$ obtained by applying the Krylov diagonalization procedure satisfies the bound
\begin{equation}
        \tilde{E}-E_0 \leq \xi\,,
\end{equation}
with probability $1-\delta$, where
\begin{align}
    \xi = \frac{\chi}{|\gamma^\prime_0|^2} + \frac{6||H||}{|\gamma_0^\prime|^2} & \left(\frac{2\chi}{\Delta^\prime} + \zeta + 8\left(1+\frac{\pi\Delta^\prime}{4||H||}\right)^{-2d+1}\right)
\end{align}
and 
\begin{align}
    \chi&\leq 2\varepsilon_Q ||H||,\\
    \zeta &\leq 2d(\varepsilon_R + \varepsilon_Q),\\
    |\gamma^\prime_0|^2 & \geq |\gamma_0|^2 - 2\varepsilon_R - 2\varepsilon_Q\,,
\end{align}
with
\begin{equation}
    \varepsilon_Q = d(d-1)t\lambda\left(\frac{t(d-1)\lambda}{N} + \sqrt{\frac{11\ln(2^{n+1}/\delta)}{NN_r}}\right)\,.
\end{equation}
Here, $\varepsilon_R$ is a regularization threshold~\cite{kirby2024analysis}, $|\gamma_0|^2$ is the overlap between $|\psi_0\rangle$ and the true ground state, $\Delta^\prime = \Delta - \chi/|\gamma^\prime_0|^2$ is a rescaled version of the spectral gap $\Delta = E_1-E_0$ and the evolution time is set to $t=\pi / (E_{2^n - 1}-E_0)$.
\end{lemma}

We follow the notation and proof structure of Ref.~\cite{kirby2024analysis}, and specialize the argument to the case where time evolution is simulated with the qDRIFT algorithm.

Let 
\begin{equation}
    \matV = [\mathrm{e}^{-\mathrm{i}\frac{(d-1)}{2}Ht}\ket{\psi_0},\mathrm{e}^{-\mathrm{i}\frac{(d-3)}{2}Ht}\ket{\psi_0},\ldots,\mathrm{e}^{\mathrm{i}\frac{d-1}{2}Ht}\ket{\psi_0}]\,,
\end{equation}
and suppose that
\begin{equation}
    (\matH, \matS) = (\matV^\dagger H \matV, \matV^\dagger \matV)
\end{equation}
are approximated via qDRIFT yielding the faulty matrices $(\matH^\prime, \matS^\prime)$. Recall that 
\begin{align}
    \matH_{ij} & = \langle \psi_0 | U^i H U^j |\psi_0\rangle = \langle \psi_0 | U^{j-i} H |\psi_0 \rangle\,,\\
    \matS_{ij} &= \langle \psi_0 | U^{j-i}|\psi_0 \rangle\,,
\end{align}
where we used the fact that $U^i$ commutes with $H$ and where $-\frac{(d-1)}{2}\leq i,j\leq \frac{d-1}{2}$. 
Then the faulty matrices $(\matH^\prime, \matS^\prime)$ can be defined as
\begin{align}\label{eq:faulty_krylov}
    \matH^\prime_{ij} &= \sum_{\vecm} p^{j-i}_{\vecm} \bra{\psi_0} V^{j-i}_{\vecm}H\ket{\psi_0} = \bra{\psi_0} \mathbb{E}[V^{j-i}]H\ket{\psi_0}\,,\notag\\
    \matS^\prime_{ij} &= \sum_{\vecm} p^{j-i}_{\vecm} \bra{\psi_0} V^{j-i}_{\vecm} \ket{\psi_0} = \bra{\psi_0} \mathbb{E}[V^{j-i}]\ket{\psi_0}\,.
\end{align}

Assume that $\matS^\prime$ is Hermitian, which can be enforced by using \cref{eq:faulty_krylov} to only calculate the elements on and above the diagonal, and obtaining the rest as conjugates of their transposes~\cite{kirby2024analysis}.
In this setting, we can apply the bound from Eq.~50 from Ref.~\cite{kirby2024analysis}, which requires upper bounding the quantities
\begin{equation}
    ||\matH - \matH^\prime|| \quad \text{and}\quad ||\matS-\matS^\prime||\,.
\end{equation}
We have that
\begin{align}
    ||\matS - \matS^\prime|| &\leq d||\matS - \matS^\prime||_{\max}=d\max_{i,j} |\matS_{ij}-\matS^\prime_{ij}|\\
    &\leq d\max_{i,j} ||U^{j-i}-\mathbb{E}[V^{j-i}]||\,,
\end{align}
where the last inequality uses that $||\ket{\psi_0}||=1$.
Then, using Proposition~3.2. from Ref.~\cite{Huang2021_qDRIFT-Concentration}, we get
\begin{equation}\label{eq:nonuniform}
    ||\matS - \matS^\prime|| \leq d \max_{-\frac{(d-1)}{2}\leq i,j\leq \frac{d-1}{2}} \frac{((j-i)t)^2\lambda^2}{N} =d\frac{(d-1)^2t^2\lambda^2}{N}\,.
\end{equation}

However, $\matS^\prime$ refers to the ideal qDRIFT channel with infinitely many randomizations. 
Let $\matS^{\prime\prime}$ be defined analogously replacing $\mathbb{E}[V^j]$ by $V^j$ and using a finite number of $N_r$ randomizations, i.e.\ the matrix that we can actually implement in practice.
Then, using Theorem~2 from Ref.~\cite{Kiss2023}, we get that with probability $1-\delta$,
\begin{equation}
    ||V^{j-i}-\mathbb{E}[V^{j-i}]||<t(j-i)\lambda\sqrt{\frac{11\ln(2^{n+1}/\delta)}{NN_r}}\,.
\end{equation}
Therefore, in practice we have that with probability at least $1-\delta$,
\begin{equation}
    ||\matS - \matS^{\prime\prime}||\leq d(d-1)t\lambda\left(\frac{t(d-1)\lambda}{N} + \sqrt{\frac{11\ln(2^{n+1}/\delta)}{NN_r}}\right)\eqqcolon \varepsilon_Q\,.
\end{equation}
The bound for $||\matH-\matH^\prime||$ can be obtained similarly using that $|\bra{\psi_0}(U^j-\mathbb{E}[V^j]H\ket{\psi_0}|\leq ||U^j-\mathbb{E}[V^j]||\cdot ||H||$. 
Resulting in the bound, with probability at least $1-\delta$,
\begin{equation}
    ||\matH - \matH^{\prime\prime}|| \leq \varepsilon_Q||H||\,.
\end{equation}

Next, using the notation from Ref.~\cite{kirby2024analysis}, we have
\begin{align}
    \chi&\coloneqq ||\matH - \matH^\prime|| + ||\matS - \matS^\prime|| \cdot ||H|| \leq 2\varepsilon_Q ||H||\,,\\
    \zeta &\coloneqq 2d(\varepsilon_R + ||\matS - \matS^\prime||)\leq 2d(\varepsilon_R + \varepsilon_Q)\,,\\
    |\gamma^\prime_0|^2 &\coloneqq |\gamma_0|^2 -2\varepsilon_R -2||\matS - \matS^\prime|| \geq |\gamma_0|^2 - 2\varepsilon_R - 2\varepsilon_Q\,,
\end{align}
where the last inequality assumes $|\gamma_0|^2 \geq 2||\matS - \matS^\prime||$ and $\varepsilon_R > 0$ denotes the regularization threshold.
The bounds hold with probability at least $1-\delta$ owing to the stochastic sampling inherent in qDRIFT.

\subsection{Proof of \texorpdfstring{\cref{thm:sqdrift_main}}{Theorem X}}\label{appendix:proof_sqdrift_main}
We follow closely the proof given in Ref.~\cite{Yu2025_SKQD} for SKQD.

\textbf{Step 1.} \cref{thm:KQD_QDRIFT} together with Lemma~1 from Ref.~\cite{Yu2025_SKQD} imply that
we can find a state $|\psi\rangle$ approximating the exact ground state as
\begin{equation}
    ||\ket{\psi} - \ket{\phi_0}||^2 \leq \tilde{\xi}= O\left(\frac{\xi}{\Delta}\right)\,.
\end{equation}

\textbf{Step 2.} Next, Lemma~2 from Ref.~\cite{Yu2025_SKQD} tells us that if $\ket{\phi_0}$ exhibits $(\alpha_L^{(0)},\beta_L^{(0)})$ concentration, then
    $\ket{\psi}$ is $(\alpha_L,\beta_L)$-concentrated with
\begin{equation}
    \alpha_L = \alpha_L^{(0)} -2\sqrt{\tilde{\xi}}\quad\text{and}\quad \beta_L = \beta_L^{(0)} -2\sqrt{\tilde{\xi}}\,.
\end{equation}

\textbf{Step 3.} We define $p^{(k)}(b_j)\coloneqq|\langle b_j|\psi^k\rangle|^2$ for each $k=0,\ldots,d-1$.
Then, by Lemma~3 from Ref.~\cite{Yu2025_SKQD} we have that for each $1 \leq i\leq L$ there exists a Krylov state $k$ such that 
\begin{equation}\label{eq:pk_bj}
    |p^{(k)}(b_i)|\geq \frac{|\gamma_0|^2\beta_L}{d^2}\,.
\end{equation}
Therefore, in the remainder of the proof we usually omit the Krylov index $k$ and assume we are working with the state such that \cref{eq:pk_bj} holds, this in turn means relabeling
\begin{align}
    U^k \rightarrow U, \quad V^k_{\vecm} \rightarrow V_{\vecm}, \quad V^k \rightarrow V \quad \text{and}\quad p^{(k)}_{\vecm}\rightarrow p_{\vecm}\,.
\end{align}

\textbf{Step 4.} Fix a bitstring $b_i$. We now prove that, with probability at least $1-\delta$ over the random choice of a single qDRIFT sequence $\vecm$, $|\sqrt{p(b_i)}-\sqrt{\tilde{p}_{\vecm}(b_i)}|\leq \varepsilon$, where 
\begin{equation}
    \varepsilon = (d-1)t\lambda\left(\frac{(d-1)t\lambda}{N} + \sqrt{\frac{11\ln(2^{n+1}/\delta)}{N}}\right)\,,
    \label{eq:qdrift-epsilon_app}
\end{equation}
and where $\sqrt{\tilde{p}_{\vecm}(b_i)}$ is the approximation of $\sqrt{p(b_i)}$ obtained for a sampled $V_{\vecm}$. 

We first note that we can bound the error on the amplitudes $|\sqrt{p(b_i)}-\sqrt{\tilde{p}_{\vecm}(b_i)}|$ by bounding on their respective operators
\begin{equation}
    |\sqrt{p(b_i)}-\sqrt{\tilde{p}_{\vecm}(b_i)}| \leq ||\ket{\psi} - \ket{\tilde{\psi}_{\vecm}}|| \leq ||U - V_{\vecm}||\,.
\end{equation}

In the remainder, we will bound the different terms of the expression
\begin{equation}\label{eq:goal_bound}
    ||U - V_{\vecm}|| \leq ||U - \mathbb{E}[V]|| + ||\mathbb{E}[V] - V_{\vecm}||\,.
\end{equation}
Recall that $U=U^k$ evolves for a time $kt$, with $k\in\{0,\ldots,d-1\}$; we therefore bound both terms for the worst case $k=d-1$.

\paragraph{First term}
From Proposition~3.2 in Ref.~\cite{Huang2021_qDRIFT-Concentration} (Eq.~95), we have that
\begin{equation}\label{eq:first_term}
    ||U - \mathbb{E}[V]||\leq \frac{(d-1)^2 t^2 \lambda^2 }{N}\,.
\end{equation}

\paragraph{Second term}
We now use Proposition~3.3 from Ref.~\cite{Huang2021_qDRIFT-Concentration}, and get that with probability $1-\delta$,
\begin{equation}\label{eq:third_term}
    ||V_{\vecm}-\mathbb{E}[V]||<(d-1)t\lambda\sqrt{\frac{11\ln(2^{n+1}/\delta)}{N}}\,.
\end{equation}

Putting all together, we can bound $|\sqrt{p(b_i)}-\sqrt{\tilde{p}_{\vecm}(b_i)}|$ with probability $1-\delta$ as
\begin{equation}
    |\sqrt{p(b_i)}-\sqrt{\tilde{p}_{\vecm}(b_i)}|\leq (d-1)t\lambda\left(\frac{(d-1)t\lambda}{N} + \sqrt{\frac{11\ln(2^{n+1}/\delta)}{N}}\right)\,.
\end{equation}

In turn, this implies that with probability at least $1-\delta$, $|\sqrt{\tilde{p}_{\vecm}(b_i)}|\geq\frac{|\gamma_0|\sqrt{\beta_L}}{d} - \varepsilon$. Therefore, 
\begin{equation}\label{eq:tildep_bound}
1-\tilde{p}_{\vecm}(b_i)\leq1-\left(\frac{|\gamma_0|\sqrt{\beta_L}}{d} - \varepsilon\right)^2\eqqcolon 1-p
\end{equation} 
with probability at least $1-\delta$.

\textbf{Step 5.} We can now bound the probability of missing the bitstring $b_i$ from $N_r$ qDRIFT realizations from which we take $S$ samples.
Note that while Ref.~\cite{Kiss2023} do not look into taking several samples per randomization, a qDRIFT channel with $N_r$ randomizations and $S$ samples per randomization is equivalent to their definition of the qDRIFT channel with $N_r$ individual experiments (Eq.~7 in~\cite{Kiss2023}):
\begin{equation}
    \frac{1}{N_r S}\sum_{r=1}^{N_r} \sum_{s=1}^S \left[V_{\vecm_r}\rho V^\dagger_{\vecm_r}\right]=\frac{1}{N_r}\sum_{r=1}^{N_r} \left[V_{\vecm_r}\rho V^\dagger_{\vecm_r}\right]=\mathcal{E}(t;N,N_r)\,.
\end{equation}
Let $X(b_i)$ denote the random variable that is $1$ if $b_i$ is sampled and $0$ otherwise, i.e.\
\begin{equation}
    \mathbb{P}[X(b_i)=1] = \sum_{\vecm} p_{\vecm}\tilde{p}_{\vecm}(b_i)\,,
\end{equation}
where ${\vecm}=(m_1,\ldots,m_N)$ ranges over all possible randomizations and $p_{\vecm}$ denotes the probability of sampling $V_{\vecm}$.
Then, the probability of not measuring $b_i$ when taking $S$ samples from each of the $N_r$ randomizations is
\begin{equation}
    p_{\text{fail}}(b_i)\coloneqq \mathbb{P}[X(b_i)=0] = \left(\sum_{\vecm}p_{\vecm}\left(1-\tilde{p}_{\vecm}(b_i)\right)^S\right)^{N_r}\,.
\end{equation}
To bound the sum inside the parenthesis, split it over the sequences $\vecm$ for which \cref{eq:tildep_bound} holds -- with total probability mass at least $1-\delta$, and for which $(1-\tilde{p}_{\vecm}(b_i))^S\leq(1-p)^S$ -- and those for which it does not, with probability mass at most $\delta$, for which we bound $(1-\tilde{p}_{\vecm}(b_i))^S\leq 1$ trivially. Using that $\sum_{\vecm}p_{\vecm}=1$ and that the $N_r$ realizations are independent, this gives
\begin{equation}
    p_{\text{fail}}(b_i) \leq \left((1-\delta)\left(1-p\right)^S + \delta\right)^{N_r}\,.
\end{equation}
Finally, applying the union bound over the $L$ important bitstrings, the overall failure probability incurs an extra factor of $L$. 

\textbf{Step 6.} We conclude applying Lemma~5 from Ref.~\cite{Yu2025_SKQD}. The state $|\tilde{\phi}\rangle = (1/C)\sum_{j=0}^{L-1}a_j|b_j\rangle$ with $C = \sqrt{\sum_{j=0}^{L-1}|a_j|^2}$, representing the restriction on the $L$-dimensional basis of important bitstrings of the target ground state defined on the full $n$-qubit Hilbert space, satisfies
\begin{equation}
    \langle \tilde{\phi}|H |\tilde{\phi}\rangle - \langle \phi_0|H |\phi_0\rangle \leq \sqrt{8}||H||\left(1-\sqrt{\alpha_L^{(0)}}\right)^{1/2}\,.
\end{equation}
Therefore, it represents a valid solution to the ground state approximation problem.

\section{Full Configuration Interaction Monte Carlo calculations}

We report in \cref{fig:convergence-fciqmc} the convergence of the FCIQMC calculation with the number of walkers. To understand the convergence trend of FCIQMC, we executed an extensive analysis using the smaller STO3G basis set.
The simulations were run within the adaptive-shift variant~\cite{Ghanem2020_AdaptiveShift-FCIQMC}, using an offset of $0.2$. These calculations yielded an energy of $-905.227(1)~\mathrm{Ha}$. Truncated CIQMC calculations at the CISDTQ and CISDTQ56 level yield $-905.1974~\mathrm{Ha}$ and $-905.222(1)~\mathrm{Ha}$ respectively. Even with hexatuple excitations, therefore, chemical accuracy (i.e.\ within $1.6~\mathrm{mHa}$) is not achieved. A slow convergence with excitation level is likely an indication of size-inconsistency error, which is a major challenge for truncated CI methods in relatively extended systems, even if they are not strongly correlated. This is corroborated by coupled-cluster calculations, which yield quite rapid convergence to the estimated FCI result: CCSD: $-905.21508~\mathrm{Ha}$, CCSDT: $-905.23194~\mathrm{Ha}$, CCSDTQ: $-905.22771~\mathrm{Ha}$, CCSDTQP: $-905.22732~\mathrm{Ha}$. Note, however, that the convergence of CC series is not monotonic. It is also noteworthy that the HCI with $10^8$ determinant variational space matches only the CCSD energy, again another indication of the slow convergence of the energy with size of variational space.

Notably, the energy obtained using FCIQMC with $1$ million determinants is lower than that obtained with HCI with approximately $100$ million determinants.
When comparing HCI and FCIQMC results, it should however be recalled that subspace size and number of walkers are not equivalent.
In fact, the number of unique determinants occupied instantaneously in FCIQMC is usually lower than the number of walkers, as many walkers will be associated to the reference determinants.
Hence, the subspace dimension of FCIQMC is effectively lower than the number of walkers.
Still, the FCIQMC energy is lower than the HCI one because, due to the stochastic walker fluctuation that is inherent to FCIQMC, the latter method (partially) captures dynamical correlation effects. For this reason, FCIQMC yields a lower energy estimate.

\section{On the SqDRIFT sampling}\label{appendix:sampling-behavior}
\begin{figure*}
\subfloat[\label{fig:FCIvsSqDRIFT}]{%
\includegraphics{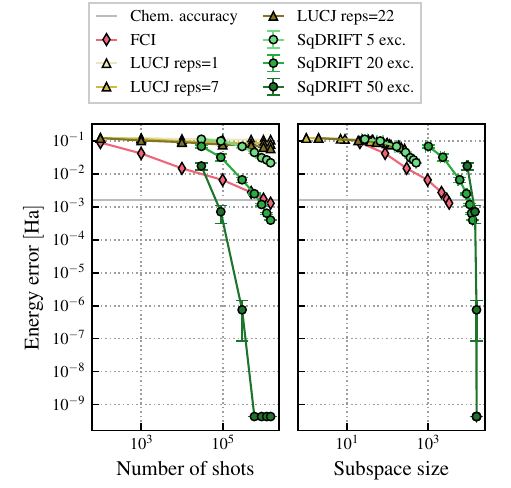}
}
\subfloat[\label{fig:ExcitationAmplitude}]{%
\includegraphics{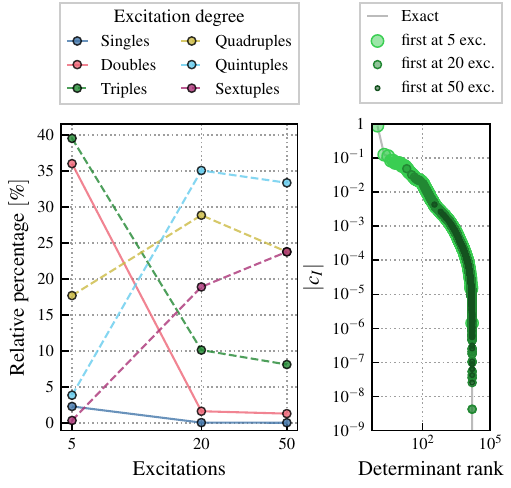}
}
\caption{\textbf{Energy convergence and excitation-degree composition of the SqDRIFT subspace for naphthalene (simulation).} (a) Absolute energy error $|E-E_{\mathrm{FCI}}|$ with respect to the exact $(10e,10o)$ full-CI energy, shown as a function of the number of measurement shots (left) and of the resulting subspace dimension (right). FCI sampling (red) is compared with SqDRIFT at $5$, $20$ and $50$ excitations (green shades) and with LUCJ at $1$, $7$ and $22$ repetitions (yellow shades). Markers and error bars denote the mean and standard deviation over independent randomizations, and the horizontal bars in the right panel give the corresponding spread in subspace size. The gray line marks chemical accuracy ($1.6\times10^{-3}~\mathrm{Ha}$). At a given subspace dimension, SqDRIFT with sufficiently many excitations reaches errors well below chemical accuracy and several orders of magnitude smaller than FCI sampling or LUCJ, which plateau at a higher error. (b) \emph{Left:} relative population (in \%) of each excitation class (singles through sextuples, classified by the number of electrons promoted with respect to the restricted Hartree-Fock reference) among the distinct determinants sampled into the SqDRIFT subspace, as a function of the number of fermionic excitations; each point aggregates $500$ randomized circuits. Singles and doubles (the CISD space) are shown with solid lines and higher excitations with dashed lines. For a low number of excitations the subspace is dominated by low-order, CISD-like excitations, whereas by $\simeq20$-$50$ excitations its weight shifts to higher classes, indicating that SqDRIFT rapidly incorporates determinants well beyond the CISD space. \emph{Right:} magnitude of the exact full-configuration-interaction coefficients $|c_I|$, with determinants ordered by decreasing weight (rank $1$ is the largest $|c_I|$). The gray curve is the exact amplitude profile over all determinants; colored markers denote the determinants captured by the SqDRIFT subspace, each shaded by the smallest number of excitations that first samples it (light to dark for $5,20,50$ excitations). SqDRIFT preferentially collects the determinants carrying the largest weight, and increasing the number of excitations extends the captured set further into the low-amplitude tail, while determinants of negligible weight remain outside the subspace.
}
\label{fig:FCIAnalysis}
\end{figure*}
\begin{figure*}
\includegraphics{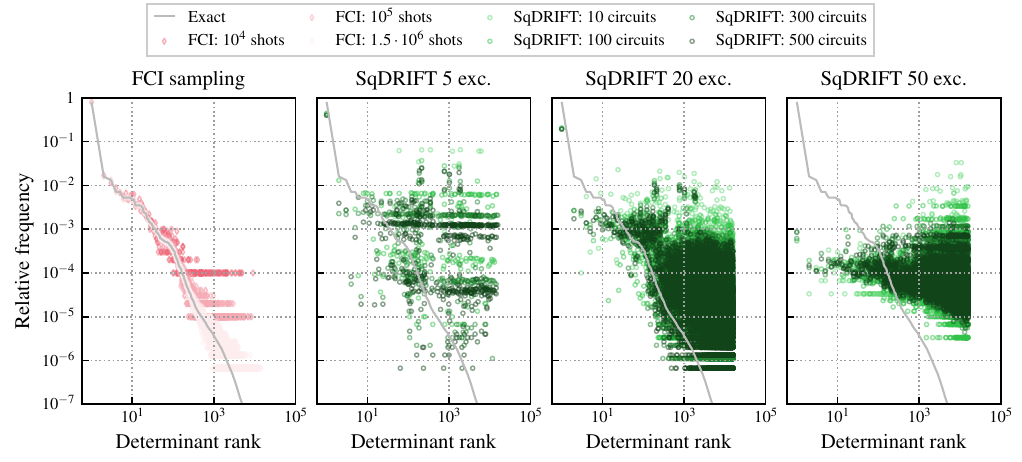}
\caption{\textbf{Sampled determinant distributions versus the exact wavefunction of naphthalene (simulation).} Relative sampling frequency of each determinant, ordered by its exact full-CI weight (rank $1$ is the largest $|c_I|$). The gray curve is the exact weight profile. The left panel shows FCI sampling (red) at increasing shot budgets ($10^{4}$, $10^{5}$ and $1.5\times10^{6}$ shots); the three remaining panels show SqDRIFT at $5$, $20$ and $50$
excitations (green), each at increasing numbers of randomized circuits ($10$ to $500$, see legend). Both samplers concentrate their counts on the high-weight, low-rank determinants. FCI
sampling tracks the exact profile down to a resolution floor set by the finite shot budget (the horizontal shelves), below which rare determinants cannot be resolved. SqDRIFT instead broadens its coverage as the number of excitations increases: at $5$ excitations the sampling is nearly uniform over a limited set of reachable determinants (flat band), whereas at $50$ excitations
the qDRIFT concentration effect re-weights the counts toward the true importance ordering, and the distribution both extends further into the tail and slopes to follow the exact profile.
}
\label{fig:SampledDistributions}
\end{figure*}
In this Section we report additional numerical analyses of the sampling behavior of SqDRIFT for naphthalene, which characterize the determinant-selection step underlying the SQD post-processing.

Before presenting the numerical analysis, we make precise the sense in which a flat distribution over the relevant configurations is the ideal sampling target. Suppose the ground state is supported, up to the target accuracy, on $N_D$ determinants, and consider a circuit that samples uniformly over this set. Recovering the entire support is then a coupon-collector problem, for which $\mathcal{O}(N_D\log N_D)$ samples suffice on average to collect all $N_D$ configurations. By contrast, sampling from the ground-state distribution $P(x)=|\langle x|\Psi_0\rangle|^2$, which is typically strongly peaked around the reference determinant, spends most of the budget re-drawing the dominant configurations: the expected number of samples needed to observe the least-probable relevant determinant scales as $1/p_{\min}$, with $p_{\min}$ the smallest weight in the support, which can be orders of magnitude larger than $N_D\log N_D$. This is the origin of the shot-budget-limited resolution floor observed for ground-state sampling in \cref{fig:SampledDistributions}, and it identifies reduced peakedness -- rather than fidelity to the ground-state distribution -- as the property that makes a sampling circuit efficient for subspace construction.

\cref{fig:FCIvsSqDRIFT} shows the absolute energy error with respect to the exact (10e, 10o) full-CI energy, as a function of the number of measurement shots (left) and of the resulting subspace dimension (right). Alongside SqDRIFT with $5$, $20$ and $50$ excitations, we report the results obtained by sampling determinants directly from the exact full-CI ground-state wave function (``FCI sampling''), and those obtained from the LUCJ ansatz with $1$, $7$ and $22$ repetitions. Since the dimension of the diagonalization subspace determines the cost of the classical eigensolver, which is common to all these approaches, the right panel isolates the quality of the selected configurations: at a given subspace dimension, SqDRIFT with a sufficient number of excitations attains energy errors well below chemical accuracy and several orders of magnitude smaller than both FCI sampling and LUCJ, which plateau at higher errors.

\cref{fig:ExcitationAmplitude} resolves the composition of the sampled subspace. The left panel reports the relative population of each excitation class -- from singles to sextuples, classified by the number of electrons promoted with respect to the restricted Hartree-Fock reference -- among the distinct determinants collected. For a small number of excitations the subspace is dominated by low-order, CISD-like configurations, whereas at $20$--$50$ excitations its weight shifts towards higher excitation classes, showing that SqDRIFT rapidly incorporates determinants well beyond the CISD space. The right panel shows the exact full-CI amplitude profile $|c_I|$, with determinants ordered by decreasing weight, and marks those captured by the SqDRIFT subspace: the protocol preferentially collects the determinants carrying the largest weight, and increasing the number of excitations extends the captured set further into the low-amplitude tail, while determinants of negligible weight remain outside the subspace.

Finally, \cref{fig:SampledDistributions} compares the sampled distributions directly with the exact one, by plotting the relative sampling frequency of each determinant against its exact full-CI rank. Sampling from the exact ground state reproduces the exact profile only down to a resolution floor set by the finite shot budget -- visible as horizontal shelves -- below which rare determinants cannot be resolved, irrespective of their chemical relevance. SqDRIFT instead generates a distribution that is markedly less peaked than the ground state: at $5$ excitations the sampling is nearly uniform over a limited set of reachable determinants, while increasing the number of excitations broadens the coverage and progressively re-weights the counts towards the true importance ordering, extending further into the tail and approaching the exact profile.

These results indicate that the benefit of SqDRIFT relative to direct ground-state sampling originates from the reduced peakedness of the sampled distribution, which lowers the sampling cost required to resolve the tail of the wave function.

\section{Additional data}\label{appendix:additional-data}
In this Section we report, for completeness, our earlier hardware results for coronene, which predate the refined protocol used for the results presented in \cref{sec:results} and were run using the IBM Quantum Heron r2 processor, \texttt{ibm\_aachen}. \cref{fig:coronene-appendix} shows the absolute ground-state energy of coronene as a function of the diagonalization subspace dimension, comparing SqDRIFT and the LUCJ ansatz against the classical references, together with the Ext-SqDRIFT extension.
\begin{figure*}
\includegraphics{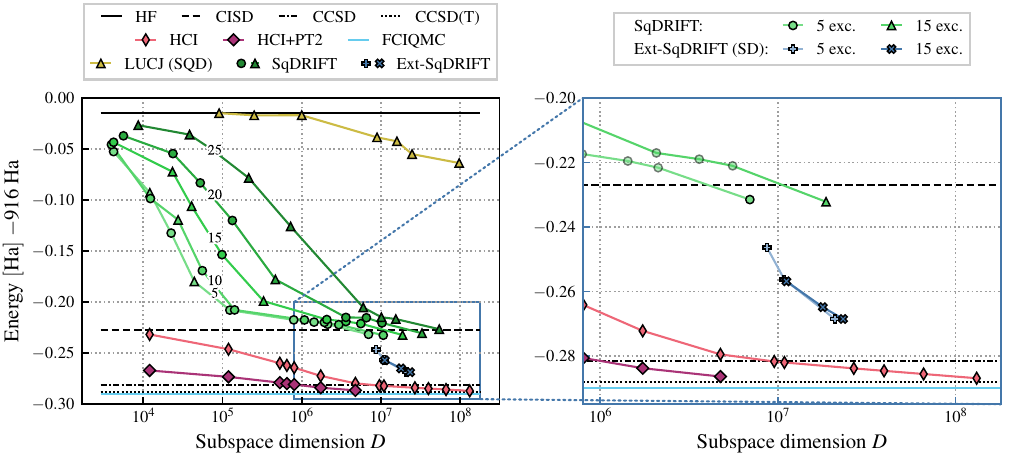}
  \caption{\textbf{Energy of coronene as a function of the subspace diagonalization dimension (hardware)}.
  (a) We plot the absolute energy values obtained from different approximate methods.
  The following classical computations serve as references: Hartree-Fock (black solid; $-916.015~\mathrm{Ha}$), CISD (black dashed; $-916.227~\mathrm{Ha}$), CCSD (black dash-dotted; $-916.282~\mathrm{Ha}$), CCSD(T) (black dotted; $-916.288~\mathrm{Ha}$), FCIQMC (cyan solid; $-916.290~\mathrm{Ha}$), HCI (red; min: $-916.287~\mathrm{Ha}$) and HCI with a deterministic PT2 correction (purple; min: $-916.286~\mathrm{Ha}$).
  For LUCJ (yellow; min: $-916.064~\mathrm{Ha}$) and SqDRIFT (shades of green; min: $-916.232~\mathrm{Ha}$) we plot the minimum energy values obtained throughout $2$ iterations of SQD post-processing including configuration recovery (CR).
  Triangles are indicative of the minimum energy obtained \emph{with} CR, circles indicate the minimum energy obtained \emph{without} CR. (b) Zoomed inset obtained from the blue box on the left, highlighting the SqDRIFT energies under CISD and the solutions found with Ext-SqDRIFT for $5$ (plus sign) and $15$ (crosses) excitations.
  }
\label{fig:coronene-appendix}
\end{figure*}
\begin{figure*}
\subfloat{%
\includegraphics{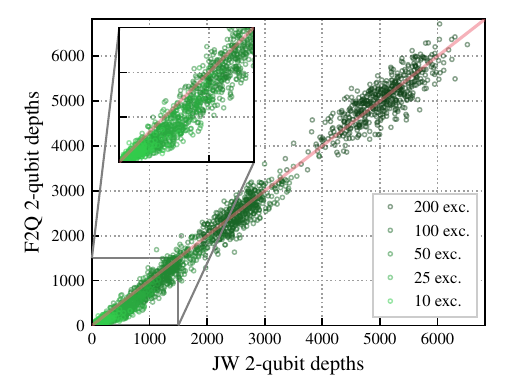}
}
\subfloat{%
\includegraphics{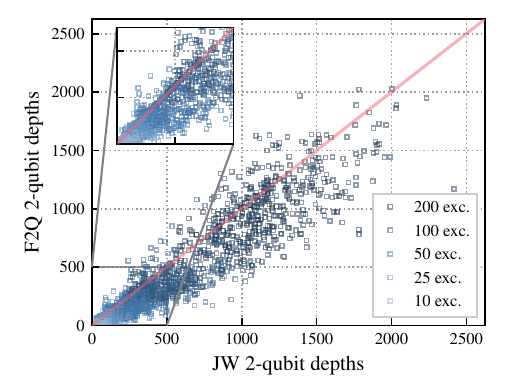}
}
\caption{\textbf{Two-qubit circuit depth improvement due to the fermion-to-qubit layout optimization.}
We plot the two-qubit gate depth of the transpiled SqDRIFT circuits (for $k=1$; these results are independent of $k$) without and with the fermion-to-qubit layout optimization along the $x$ and $y$ axis, respectively.
Thus, if a data point lies \emph{below} the (red) diagonal, the F2Q layout optimization has decreased the two-qubit gate circuit depth.
The two panels show the results obtained for the numerical simulations of naphthalene for both the HF (a) and LO (b) orbitals, respectively. As expected, the two-qubit gate circuit depths are larger in the HF orbitals by about a factor of $\sim4$ compared to the LO orbitals. The number of excitations included in each randomized circuit is indicated by the color shade.
}
\label{fig:F2QOpt}
\end{figure*}
We complement these classical energy computations with those computed by SQD with the LUCJ ansatz as well as SqDRIFT.
The LUCJ ansatz had a two-qubit gate depth of $98$ and count of $1728$.
The specifications for the randomized SqDRIFT circuits vary similarly for all values of $k$.
Here, we use $5$ different groups of circuit randomizations ranging from $5$ to $25$ excitations.
For $k=1$ these circuits have a two-qubit gate depth (count) of $36\pm33$ ($60\pm51$) and $777\pm526$ ($2089\pm1584$) for $5$ and $25$ excitations, respectively.
More details are listed in \cref{tab:hw_cor} and \cref{fig:two_qubit_hists}. 

As shown in \cref{fig:coronene-appendix}, the SQD energy obtained with the LUCJ ansatz slowly improves upon the base HF energy.
As indicated by the triangles rather than circles, the minimal energy throughout the SQD pipeline is obtained at an iteration that leverages configuration recovery. 

In contrast, the results obtained from SqDRIFT show faster convergence than those obtained from the LUCJ circuit, even surpassing CISD without the need for configuration recovery.
We expect that samples obtained from the lower-depth circuits, which are shallower and therefore less susceptible to noise, will contribute the most to this energy improvement.
Interestingly, the SqDRIFT energy increases, for a given $D$ value, by increasing the number of excitations involved in their circuits.
This trend is opposite to the one observed, for instance, for noiseless simulations of naphthalene where, at a given subspace dimension, the circuits involving the largest number of excitations yield the lowest energy estimates.
This behavior can be rationalized as follows: circuits constructed with low number of excitations (i.e.\, shallower circuits) are less affected by noise, but they explore only low-excited determinants.
For this reason, the corresponding SqDRIFT energy is slightly lower than the CISD reference energy.
When a larger number of excitations is used, the relevant higher-order excitations are sampled as well.
For a fixed sampling budget, this should result in better energy estimates, as predicted by the SqDRIFT asymptotic error bounds.
However, this trend is not observed in the experiments because hardware noise introduces errors in the sampled bitstrings, irrespectively of their excitation degree. We note that, because the final energy is obtained by variational diagonalization in the collected subspace, the depth-dependent noise of the randomized ensemble affects the \emph{coverage} of relevant configurations rather than introducing a weight-dependent bias in the energy estimate: every collected determinant enters the diagonalization on an equal footing, independently of the depth or noise of the circuit that produced it. The increase in the number of two-qubit gates with the number of excitations is shown in \cref{fig:two_qubit_hists}. We also see that the number of two-qubit gates of the LUCJ circuit is significantly higher than that of the shallower SqDRIFT circuits, yielding noisier measurement outcomes, likely being the cause of the significantly inferior performance of the LUCJ circuit as compared to the SqDRIFT ones. 
Furthermore, as energies are computed and compared using a fixed subspace dimension, including higher-order excitations reduces the number of double excitations included in the subspace.
Since for coronene double excitations contribute most significantly to the correlation energy, achieving higher accuracy likely requires a more careful balance between the excitation degree and the allocation of sampling resources.
For this reason, the energy estimate remains higher than the CISD reference.

This effect is partially mitigated by configuration recovery (as indicated by the triangles on the curves of $15$, $20$, and $25$ excitations).
The energy estimate can be improved by combining samples obtained from circuits using few and many excitations (dark green data points in \cref{fig:coronene-appendix}) since the space of excited Slater determinants can be explored more systematically.
Note, however, that this step is heuristic, and is not guaranteed to improve the SqDRIFT accuracy in general.

Finally, in the zoomed-in inset of \cref{fig:coronene-appendix} we include also the results obtained by applying single, double, and triple excitation operators to the configuration obtained by SqDRIFT using a randomized circuit group made with $15$ excitations.
With the addition of a single extra computing step, we see an improvement of $\sim 35~\mathrm{mHa}$, bringing the SqDRIFT energies to within $20~\mathrm{mHa}$ of the best HCI result, which requires an order of magnitude more electronic configurations.
\begin{figure*}
\subfloat[Naphthalene circuits\label{fig:2q_count_nap}]{%
\includegraphics{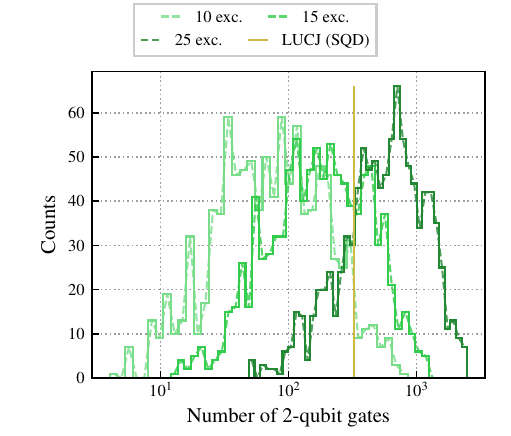}
}
\subfloat[Coronene circuits\label{fig:2q_count_cor}]{%
\includegraphics{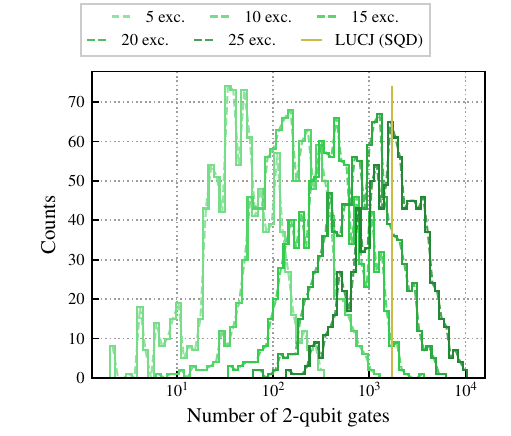}
}
\caption{\textbf{Histogram of the number of two-qubit gates for the hardware experiments of (a) naphthalene (\cref{fig:Naphthalene}) and (b) coronene (\cref{fig:coronene-appendix}).} The statistics of these histograms are summarized on \cref{tab:hw_nap} and \cref{tab:hw_cor}.}
\label{fig:two_qubit_hists}
\end{figure*}
\section{Computational details}\label{appendix:computational-details}

The following software packages have been used in this work.
Qiskit SDK \texttt{v2.1.0}~\cite{Qiskit2024} for the implementation of the qDRIFT sampling and generation of the quantum circuits.
\texttt{rustiq v0.0.8}~\cite{Debrugiere2024_rustiq} was used to optimize the depth of the quantum circuits.
\texttt{qiskit-addon-sqd v0.11.0}~\cite{qiskit-addon-sqd} in combination with the latest development version of \texttt{qiskit-addon-dice-solver} (git commit: \texttt{785f8b4})~\cite{qiskit-addon-dice-solver} leveraging a development version of \texttt{DICE} (git commit: \texttt{3198db1})~\cite{Holmes2016_HBCI, holmes2016efficient, Sharma2017_Semistochastic} were used for the SQD post-processing step.
The quantum circuits were executed on the \texttt{ibm\_aachen} Heron r2 chip and \texttt{ibm\_boston} Heron r3 chip by IBM Quantum, both providing a total of 156 superconducting qubits.
We report statistics on hardware resources of the randomized SqDRIFT circuits executed on hardware in \cref{tab:hw_nap,tab:hw_cor} and \cref{fig:two_qubit_hists}, as well as an example in \cref{fig:F2QOpt} (for naphthalene and both orbital choices) of the circuit depth reduction that can be obtained using the F2Q protocol outlined in \cref{sec:F2Q}. For both cases, the two-qubit gate count increases with the excitation number. Given typical two-qubit error rates, circuits with more than $\sim 10^3$ two-qubit gates enter a regime where sampling quality is dominated by hardware noise. Consistent with this, we observe a pronounced degradation from $>15$ excitations, in line with the results in \cref{fig:coronene-appendix}. Over the same range, SqDRIFT uses the available depth budget more efficiently than the LUCJ ansatz, achieving similar expressivity at lower depth and therefore reduced noise accumulation. We emphasize that this threshold is diagnostic rather than operational: no circuits are excluded on the basis of their depth, and all sampled configurations from every randomization enter the configuration-recovery and diagonalization steps.
The hardware experiments were managed via the quantum spank plugin for Slurm \texttt{v0.1.0}~\cite{Sitdikov2025_quantum}.
For all instances of the LUCJ ansatz, the circuits are constructed by fixing the parameters to those obtained from the $t_1$ and $t_2$ amplitudes of a prior CCSD calculation of the respective system run with \texttt{ffsim v0.0.57}~\cite{ffsim}.

HF and CISD calculations were run using the PySCF package~\cite{Sun2020_PySCF} (version 2.4.0).
The active orbitals were selected using the AVAS protocol~\cite{Sayfutyarova2017_AVAS}, taking the $p_z$ orbitals as the reference atomic orbitals.
For both naphthalene and coronene, all calculations use the cc-pVDZ basis set.
HCI calculations were run using the default configuration of \texttt{DICE}.
The reference geometry of naphthalene was taken from Ref.~\cite{Ghosh2017_GAS-pDFT}, and the one of coronene was taken from Ref.~\cite{Blunt2021_FCIQMC-Acene}.
\begin{table}[t]
    \subfloat[Two-qubit gate depths]{
        \begin{tabular}{cccccccc}
            \toprule
            \# exc. & mean & std. dev. & min & $25\%$ & median & $75\%$ & max \\
            \midrule
            10 & 69.40 & 68.74 & 2 & 23.00 & 45.50 & 96.25 & 450 \\
            15 & 140.22 & 117.45 & 4 & 51.75 & 104.00 & 199.25 & 684 \\
            25 & 354.71 & 221.90 & 19 & 186.00 & 315.50 & 490.00 & 1200 \\
            \bottomrule
        \end{tabular}
    }
    \subfloat[Two-qubit gate counts]{
        \begin{tabular}{cccccccc}
            \toprule
            \# exc. & mean & std. dev. & min & $25\%$ & median & $75\%$ & max \\
            \midrule
            10 & 112.85 & 116.37 & 4 & 35.00 & 72.00 & 152.00 & 855 \\
            15 & 246.46 & 217.23 & 12 & 91.75 & 174.00 & 349.25 & 1420 \\
            25 & 691.73 & 467.31 & 49 & 341.00 & 591.00 & 928.25 & 2499 \\
            \bottomrule
        \end{tabular}
    }
    \caption{\textbf{Statistics of the $1000$ randomized SqDRIFT circuits (for $k=1$; these results are independent of $k$) executed on hardware for naphthalene (\cref{fig:Naphthalene}).}
    The columns indicate the number of excitations included in each randomized circuit followed by the mean, standard deviation, minimum, $25\%$ percentile, median, $75\%$ percentile, and maximum values of the (a) two-qubit gate depth and (b) two-qubit gate count, respectively.
    }
    \label{tab:hw_nap}
\end{table}
\begin{table}[t]
    \subfloat[Two-qubit gate depths]{
        \begin{tabular}{cccccccc}
            \toprule
            \# exc. & mean & std. dev. & min & $25\%$ & median & $75\%$ & max \\
            \midrule
            5 & 36.34 & 32.51 & 2 & 15.00 & 25.00 & 52.00 & 210 \\
            10 & 102.46 & 85.65 & 2 & 39.75 & 76.00 & 141.25 & 513 \\
            15 & 227.05 & 174.65 & 14 & 92.75 & 176.00 & 313.00 & 1258 \\
            20 & 440.38 & 317.15 & 16 & 196.00 & 372.00 & 591.00 & 2373 \\
            25 & 777.01 & 526.23 & 51 & 391.00 & 665.00 & 1044.25 & 3550 \\
            \bottomrule
        \end{tabular}
    }
    \subfloat[Two-qubit gate counts]{
        \begin{tabular}{cccccccc}
            \toprule
            \# exc. & mean & std. dev. & min & $25\%$ & median & $75\%$ & max \\
            \midrule
            5 & 59.93 & 51.03 & 2 & 25.00 & 45.00 & 80.00 & 336 \\
            10 & 200.38 & 165.84 & 6 & 85.00 & 148.00 & 264.25 & 1276 \\
            15 & 507.29 & 416.14 & 34 & 210.00 & 380.50 & 664.00 & 3746 \\
            20 & 1086.24 & 869.70 & 62 & 451.00 & 841.50 & 1435.75 & 6324 \\
            25 & 2089.30 & 1583.73 & 137 & 941.25 & 1661.00 & 2856.00 & 10423 \\
            \bottomrule
        \end{tabular}
    }
    \caption{\textbf{Statistics of the $1000$ randomized SqDRIFT circuits (for $k=1$; these results are independent of $k$) executed on hardware for coronene (\cref{fig:coronene-appendix}).}
    The columns indicate the number of excitations included in each randomized circuit followed by the mean, standard deviation, minimum, $25\%$ percentile, median, $75\%$ percentile, and maximum values of the (a) two-qubit gate depth and (b) two-qubit gate count, respectively.
    }
    \label{tab:hw_cor}
\end{table}
\end{document}